\def\Submission{1}

\makeatletter
\def\input@path{{styles/}{../styles/}}
\makeatother

\ifx\Submission\undefined%
    \newcommand{\InSubm}[1]{}
    \newcommand{\InNotSubm}[1]{#1}
\else%
    \newcommand{\InSubm}[1]{#1}
    \newcommand{\InNotSubm}[1]{}
\fi

\ifx\SoCG\undefined%
    \documentclass[11pt]{article}%
    \providecommand{\SoCGVer}[1]{}%
    \providecommand{\NotSoCGVer}[1]{#1}%
    \providecommand{\RegVer}[1]{#1}%
\else%
    \documentclass[anonymous,a4paper,USenglish,autoref,thm-restate]%
    {socg-lipics-v2021}
    \providecommand{\SoCGVer}[1]{#1}%
    \providecommand{\NotSoCGVer}[1]{}%
    \providecommand{\RegVer}[1]{}%
    \def\ProofInAppendix{TRUE}
\fi

\NotSoCGVer{%
   \def\UseBibLatex{1}%
}

\ifx\sarielComp\undefined%
\newcommand{\SarielComp}[1]{}
\newcommand{\NotSarielComp}[1]{#1}%
\else
\newcommand{\SarielComp}[1]{#1}%
\newcommand{\NotSarielComp}[1]{}%
\fi

\NotSoCGVer{%
   \usepackage[in]{fullpage}%
   \usepackage{amsmath}%
   \usepackage{amssymb}%
   \usepackage[cmyk]{xcolor}%
   \usepackage{euscript}%
   \usepackage{amsmath}%
   \usepackage{marvosym}%
} \usepackage{scalerel}%

\usepackage{titlesec}%
\titlelabel{\thetitle. }%

\usepackage{algorithm}%
\usepackage[noend]{algpseudocode}

\usepackage{graphicx}%
\usepackage{xcolor}%
\usepackage{mleftright}%
\usepackage{xspace}%

\usepackage{caption}%

\SarielComp{\usepackage{sariel_colors}}%

\usepackage{proofapnd}

\newcommand{\remove}[1]{}%

\usepackage{hyperref}%
\hypersetup{%
   breaklinks,%
   colorlinks=true,%
   urlcolor=[rgb]{0.25,0.0,0.0},%
   linkcolor=[rgb]{0.5,0.0,0.0},%
   citecolor=[rgb]{0,0.2,0.445},%
   filecolor=[rgb]{0,0,0.4},
   anchorcolor=[rgb]={0.0,0.1,0.2}%
}
\usepackage[ocgcolorlinks]{ocgx2}

\NotSoCGVer{%
   \usepackage[amsmath,thmmarks]{ntheorem}%
   \theoremseparator{.}%

   \theoremstyle{plain}%
   \newtheorem{theorem}{Theorem}[section]

   \newtheorem{lemma}[theorem]{Lemma}

   \theoremstyle{plain}%
   \theoremheaderfont{\sf} \theorembodyfont{\upshape}%
   \newtheorem*{remark:unnumbered}[theorem]{Remark}%
   \newtheorem{remark}[theorem]{Remark}%
   \newtheorem{definition}[theorem]{Definition}
   \newtheorem{example}[theorem]{Example}

}

\newtheorem{defn}[theorem]{Definition}

\newcommand{\myqedsymbol}{\rule{2mm}{2mm}}

\NotSoCGVer{%
   \theoremheaderfont{\em}%
   \theorembodyfont{\upshape}%
   \theoremstyle{nonumberplain}%
   \theoremseparator{}%
   \theoremsymbol{\myqedsymbol}%
   \newtheorem{proof}{Proof:}%
}

\definecolor{nalmostblack}{rgb}{0, 0, 0.7}
\providecommand{\emphic}[2]{%
   \textcolor{nalmostblack}{%
      \textbf{\emph{#1}}}%
   \index{#2}}

\providecommand{\emphi}[1]{\emphic{#1}{#1}}

\definecolor{almostblack}{rgb}{0, 0, 0.5}
\providecommand{\emphw}[1]{{\emph{{\textcolor{almostblack}{#1}}}}}%

\numberwithin{figure}{section}%
\numberwithin{table}{section}%
\numberwithin{equation}{section}%

\newcommand{\atgen}{\symbol{'100}}
\newcommand{\SarielThanks}[1]{\thanks{Department of Computer Science;
      University of Illinois; 201 N. Goodwin Avenue; Urbana, IL,
      61801, USA; {\tt sariel\atgen{}illinois.edu}; {\tt
         \url{http://sarielhp.org/}.} #1}}

\newcommand{\HLinkY}[2]{\hyperref[#2]{#1}}
\newcommand{\HLink}[2]{\hyperref[#2]{#1~\ref*{#2}}}
\newcommand{\HLinkSuffix}[3]{\hyperref[#2]{#1\ref*{#2}{#3}}}

\newcommand{\seclab}[1]{\label{sec:#1}}
\newcommand{\secref}[1]{\HLink{Section}{sec:#1}}

\newcommand{\apndlab}[1]{\label{apnd:#1}}

\newcommand{\lemlab}[1]{\label{lemma:#1}}
\newcommand{\lemref}[1]{{\expandafter\HLink{Lemma}{lemma:#1}}}%
\newcommand{\Xlemref}[1]{\noexpand{\noexpand\HLink{Lemma}{lemma:#1}}}%

\newcommand{\figlab}[1]{\label{fig:#1}}
\newcommand{\figref}[1]{\HLink{Figure}{fig:#1}}

\newcommand{\thmlab}[1]{{\label{theo:#1}}}
\newcommand{\thmref}[1]{\noexpand\HLink{Theorem}{theo:#1}}
\newcommand{\thmrefY}[2]{\HLinkY{#2}{theo:#1}}

\providecommand{\deflab}[1]{\label{def:#1}}
\newcommand{\defref}[1]{\HLink{Definition}{def:#1}}
\newcommand{\defrefY}[2]{\hyperref[def:#2]{#1}}

\providecommand{\eqlab}[1]{}%
\renewcommand{\eqlab}[1]{\label{equation:#1}}
\newcommand{\Eqref}[1]{\HLinkSuffix{Eq.~(}{equation:#1}{)}}

\newcommand{\Set}[2]{\left\{ #1 \;\middle\vert\; #2 \right\}}
\newcommand{\pth}[2][\!]{\mleft({#2}\mright)}%
\newcommand{\pbrcx}[1]{\left[ {#1} \right]}%

\newcommand{\ExChar}{\mathbb{E}}%
\newcommand{\ExSym}{\mathop{\ExChar}}%
\newcommand{\Ex}[2][\!]{\ExSym#1\pbrcx{#2}}

\newcommand{\ProbLTR}{\mathbb{P}}%

\newcommand{\Prob}[1]{\mathop{\ProbLTR} \mleft[ #1 \mright]}%

\newcommand{\ceil}[1]{\left\lceil {#1} \right\rceil}
\newcommand{\floor}[1]{\left\lfloor {#1} \right\rfloor}

\newcommand{\brc}[1]{\left\{ {#1} \right\}}
\newcommand{\cardin}[1]{\left| {#1} \right|}%
\newcommand{\cardinT}[1]{\left| {#1} \right|_T^{}}%

\renewcommand{\th}{th\xspace}

\renewcommand{\Re}{\mathbb{R}}%

\usepackage[inline]{enumitem}

\newlist{compactenumA}{enumerate}{5}%
\setlist[compactenumA]{topsep=0pt,itemsep=-1ex,partopsep=1ex,parsep=1ex,%
   label=(\Alph*)}%

\newlist{compactenuma}{enumerate}{5}%
\setlist[compactenuma]{topsep=0pt,itemsep=-1ex,partopsep=1ex,parsep=1ex,%
   label=(\alph*)}%

\newlist{compactenumI}{enumerate}{5}%
\setlist[compactenumI]{topsep=0pt,itemsep=-1ex,partopsep=1ex,parsep=1ex,%
   label=(\Roman*)}%

\newlist{compactenumi}{enumerate}{5}%
\setlist[compactenumi]{topsep=0pt,itemsep=-1ex,partopsep=1ex,parsep=1ex,%
   label=(\roman*)}%

\newlist{compactitem}{itemize}{5}%
\setlist[compactitem]{label=\ensuremath{\bullet}}%
\setlist[compactitem]{topsep=0pt,itemsep=-1ex,partopsep=1ex,parsep=1ex,%
   label=\ensuremath{\bullet}}%

\usepackage{stmaryrd}%
\providecommand{\IntRange}[1]{\mleft\llbracket #1 \mright\rrbracket}
\newcommand{\IRX}[1]{\IntRange{#1}}%

\usepackage{wasysym}

\newcommand{\Pairs}{\Mh{\Pi}}%
\newcommand{\nG}{\Mh{\mathsf{N}}}%
\newcommand{\pr}{\Mh{\xi}}
\newcommand{\nMoat}{\Mh{n}^{}_{\Moat}}
\newcommand{\CBoxes}{\Mh{\mathcal{B}}}%
\newcommand{\bx}{\Mh{\mathsf{r}}}%
\newcommand{\cenX}[1]{\overline{\mathsf{c}}\pth{#1}}
\newcommand{\dc}{\Mh{\Delta}}

\newlength{\savedparindent}

\providecommand{\Mh}[1]{#1}%

\newcommand{\eps}{\varepsilon}

\newcommand{\etal}{\textit{et~al.}\xspace}

\newcommand{\Term}[1]{\textsf{#1}}

\newcommand{\LCA}{\Term{LCA}\xspace}
\newcommand{\tldO}{\scalerel*{\widetilde{O}}{j^2}}%
\newcommand{\tldOmg}{\scalerel*{\widetilde{\Omega}}{j^2}}%

\newcommand{\FaroukThanks}[1]{%
   \thanks{Department of Computer Science; University of Illinois; 201
      N. Goodwin Avenue; %
      Urbana, IL, 61801, %
      USA; %
      {\tt eyharb2\atgen{}illinois.edu}; %
      {\tt \url{https://farouky.github.io/}.}%
      #1}%
}

\newcommand{\pa}{\Mh{p}}%
\newcommand{\pz}{\Mh{z}}%

\newcommand{\p}{\Mh{p}}
\newcommand{\q}{\Mh{q}}

\newcommand{\pu}{\Mh{u}}%

\newcommand{\dY}[2]{\left\| #1  #2 \right\|}%
\newcommand{\dTY}[2]{\left\| #1  #2 \right\|_T}%

\newcommand{\grid}{\Mh{\mathsf{K}}}%

\providecommand{\G}{\Mh{G}}%
\providecommand{\GA}{\Mh{H}}%
\providecommand{\GX}[1]{\Mh{G}\pth{#1}}%
\renewcommand{\G}{\Mh{G}}%

\providecommand{\P}{\Mh{P}}%
\renewcommand{\P}{\Mh{P}}%
\newcommand{\PA}{\Mh{Q}}%
\newcommand{\PB}{\Mh{U}}%

\newcommand{\EG}{\Mh{E}}%
\newcommand{\EGX}[1]{\Mh{E}\pth{#1}}%

\newcommand{\cpX}[1]{\Mh{\mathrm{c{}p}}\pth{#1}}%
\newcommand{\diamX}[1]{\mathrm{diam}\pth{#1}}%

\newcommand{\spread}{\Mh{\Phi}}
\newcommand{\spreadX}[1]{\spread\pth{#1}}

\newcommand{\WeightX}[1]{\Mh{\omega} \pth{#1}}
\newcommand{\wX}[1]{\WeightX{#1}}%

\newcommand{\SaveContent}[2]{%
   \expandafter\newcommand{#1}{#2}%
}

\DeclareFontFamily{U}{BOONDOX-calo}{\skewchar\font=45 }
\DeclareFontShape{U}{BOONDOX-calo}{m}{n}{<-> s*[1.05] BOONDOX-r-calo}{}
\DeclareFontShape{U}{BOONDOX-calo}{b}{n}{<-> s*[1.05] BOONDOX-b-calo}{}
\DeclareMathAlphabet{\mathcalb}{U}{BOONDOX-calo}{m}{n}
\SetMathAlphabet{\mathcalb}{bold}{U}{BOONDOX-calo}{b}{n}
\DeclareMathAlphabet{\mathbcalb}{U}{BOONDOX-calo}{b}{n}

\providecommand{\TPDF}[2]{\texorpdfstring{#1}{#2}}

\newcommand{\Family}{\Mh{\mathcal{F}}}%

\newcommand{\Net}{\Mh{N}}%

\providecommand{\BibLatexMode}[1]{}
\providecommand{\BibTexMode}[1]{#1}

\ifx\UseBibLatex\undefined%
  \renewcommand{\BibLatexMode}[1]{}
  \renewcommand{\BibTexMode}[1]{#1}
\else
  \renewcommand{\BibLatexMode}[1]{#1}
  \renewcommand{\BibTexMode}[1]{}
\fi

\BibLatexMode{%
   \usepackage[bibencoding=utf8,style=alphabetic,backend=biber]{biblatex}%
   \usepackage{sariel_biblatex}%
}

\newcommand{\CHX}[1]{\Mh{\mathcal{C}}\pth{#1}}%

\newcommand{\VV}{\Mh{V}}%
\newcommand{\Edges}{\Mh{E}}%

\newcommand{\cell}{\Mh{\mathsf{C}}}%
\newcommand{\cellA}{\Mh{\mathsf{D}}}%

 \smallskip%

\newcommand{\ffrac}[2]{#1/#2}

\providecommand{\Mh}[1]{#1}
\newcommand{\origin}{\Mh{\mathsf{o}}}%

\newcommand{\cone}{\Mh{\mathcalb{c}}}%

\newcommand{\cen}{\Mh{z}}%

\newcommand{\EMST}{\Term{EMST}\xspace}

\newcommand{\VC}{\Term{VC}\xspace}
\newcommand{\Boruvka}{Bor\r{u}vka\xspace}

\newcommand{\hdelta}{\Mh{\varsigma}}

\newcommand{\DT}{\Mh{\mathcal{D}}}%
\newcommand{\DTX}[1]{\Mh{\mathcal{DT}}\pth{#1}}

\newcommand{\pencilX}[1]{\mathsf{pen}\pth{#1}}

\providecommand{\ring}{}
\renewcommand{\ring}{\Mh{\mathcalb{r}}}%
\newcommand{\Ring}{\Mh{\mathsf{O}}}%
\newcommand{\ball}{\Mh{\mathcalb{b}}}%
\newcommand{\ballY}[2]{\ball\pth{#1,#2}}

\newcommand{\ballZY}[2]{\ball_{#1}\pth{#2}}

\newcommand{\BallY}[2]{\Mh{\mathcal{B}}_{#1}\pth{#2}}
\newcommand{\ballYT}[2]{\ball_T\pth{#1,#2}} %

\newcommand{\circumX}[1]{\mathsf{circum}\pth{#1}}

\newcommand{\vicX}[1]{\mathsf{vic}\pth{#1}}
\newcommand{\vl}{\Mh{\varphi}}
\newcommand{\vlr}{\Mh{\delta}}
\newcommand{\bbY}[2]{\Mh{\mathsf{R}}\pth{#1, #2}}
\newcommand{\volX}[1]{\Mh{vol}\pth{#1}}

\newcommand{\MST}{\Term{MST}\xspace}%
\newcommand{\MSTP}{\Mh{\mathcal{T}}}%
\newcommand{\CC}{\Mh{\mathsf{C}}}%

\newcommand{\Cones}{\Mh{\mathcalb{C}}}%
\newcommand{\ConesX}[1]{\Cones_{#1}}%

\newcommand{\Moat}{\Mh{\mathcal{M}}}%

\newcommand{\nnX}[1]{\Mh{\mathcalb{n}}_{#1}}
\newcommand{\nnY}[2]{\Mh{\mathcalb{n}}_{#1}\pth{#2}}

\newcommand{\reachC}{\Mh{\ell}}
\newcommand{\reachX}[1]{\reachC\pth{#1}}
\newcommand{\reachY}[2]{\reachC\pth{#1,#2}}
\newcommand{\cPower}{\Mh{\kappa}}

\newcommand{\VX}[1]{\Mh{V}\pth{#1}}

\newcommand{\idX}[1]{\mathsf{id}\pth{#1}}

\newcommand{\IFF}{\iff}
\newcommand{\simplex}{\Mh{\nabla}}%
\newcommand{\pvX}[1]{\mathsf{p{}v}\pth{#1}}
\newcommand{\powX}[1]{\ceil{#1}_2}
\newcommand{\Hd}{[0,1]^d}
\newcommand{\wardX}[1]{\Mh{\varocircle}_{#1}}%
\newcommand{\starX}[1]{ \hexstar_{#1}}

\newcommand{\GC}{\Mh{G}_\angle}%
\newcommand{\GCX}[1]{\GC\pth{#1}}
\newcommand{\GCY}[2]{\Mh{G}_{\angle, #2}\pth{#1}}
\newcommand{\portalsX}[1]{\partial\pth{#1}}

\newcommand{\IR}{\Mh{\mathcal{F}}}%

\newcommand{\vSX}[1]{\Mh{\mathcalb{v}}_{#1}}

\newcommand{\D}{\Mh{\mathcal{D}}}%
\newcommand{\Arr}{\Mh{\mathop{\mathrm{\mathcal{A}}}}}%
\newcommand{\ArrX}[1]{\Arr\pth{#1}}%
\newcommand{\Mares}{Mare\v{s}\xspace}

\SoCGVer{\newcommand{\myparagraph}[1]{\subparagraph{#1.}}}
\NotSoCGVer{\newcommand{\myparagraph}[1]{\paragraph{#1.}}}

\newcommand{\GSet}{\Mh{\mathcal{C}}}%
\newcommand{\FGSet}{\Mh{C}}%
\newcommand{\aotimes}{*}

\NotSoCGVer{%
   \title{Revisiting Random Points: Combinatorial Complexity and
      Algorithms}

   \author{%
      Sariel Har-Peled%
      \SarielThanks{%
         Work on this paper was partially supported by NSF AF award
         CCF-2317241.  }%
      \and%
      Elfarouk Harb%
      \FaroukThanks{}%
   } }

\SoCGVer{%
   \title{Revisiting Random Points: Combinatorial Complexity and
      Algorithms}

   \author{Sariel Har-Peled}{Department of Computer Science;
      University of Illinois Urbana-Champaign
      \and \url{http://sarielhp.org/}}{sariel@illinois.edu}{}{}

   \author{Elfarouk Harb}{Department of Computer Science; University
      of Illinois Urbana-Champaign
      \and \url{https://farouky.github.io}}{eyfmharb@gmail.com}{}{}

   \Copyright{Sariel Har-Peled and Elfarouk Harb}
   \authorrunning{Sariel Har-Peled and Elfarouk Harb}
}

\NotSoCGVer{%
   \BibLatexMode{%
      \bibliography{rand_mst}%
   }%
}

\begin{document}

\maketitle

\begin{abstract}
    Consider a set $\P$ of $n$ points picked uniformly and
    independently from $[0,1]^d$, where $d$ is a constant. Such a
    point set is well behaved in many aspects and has several
    structural properties.  For example, for a fixed $r \in [0,1]$, we
    prove that the number of pairs of $\binom{\P}{2}$ at a distance at
    most $r$ is concentrated within an interval of length $O(n\log n)$
    around the expected number of such pairs for the torus distance.
    We also provide a new proof that the expected complexity of the
    Delaunay triangulation of $\P$ is linear -- the new proof is
    simpler and more direct than previous proofs.

    In addition, we present simple linear time algorithms to construct
    the Delaunay triangulation, Euclidean \MST, and the convex hull of
    the points of $\P$. The \MST algorithm uses an interesting
    divide-and-conquer approach.  Finally, we present a simple
    $\tilde{O}(n^{4/3})$ time algorithm for the distance selection
    problem, for $d=2$, providing a new natural justification for the
    mysterious appearance of $n^{4/3}$ in algorithms for this problem.
\end{abstract}

\section{Introduction}

\myparagraph{Input model}

Fix a constant dimension $d\geq 2$.  For
$i\in \IRX{n} = \brc{1,\ldots, n}$, uniformly and independently sample
a point $\p_i$ from $[0,1]^d$. Let $\P = \Set{\p_i }{i\in
   \IRX{n}}$. The \emphw{euclidean graph} on $\P$ is
$\GX{\P} = \bigl(\P, \binom{\P}{2} \bigr)\Bigr.$, with the edge
$\p_i \p_j$ having weight $\wX{\p_i \p_j} = \dY{\p_i}{\p_j}$, for
$\p_i, \p_j\in \P$, where $\binom{\P}{2} = \Set{p q}{p,q \in \P}$.
This graph has quadratic number of edges, but is defined by only
$O(n)$ input numbers.  Natural questions to ask about $\P$ and
$\GX{\P}$ include:
\begin{compactenumA}
    \smallskip%
    \item What is the combinatorial complexity of the
    convex-hull/Delaunay triangulation of $P$?

    \smallskip%
    \item How quickly can one compute the convex-hull/Delaunay
    triangulation/MST/etc of $P$?

    \smallskip%
    \item What is the length of the median edge in $\GX{\P}$, and how
    concentrated is this value?
\end{compactenumA}
\smallskip%
All these questions have surprisingly good answers -- linear
complexity, linear running time algorithms, and strong concentration,
respectively.  Here, we revisit these questions, presenting new
simpler proofs and algorithms for them.

\subsection{Background}

There is a lot of work in stochastic and integral geometry on
understanding the behavior of random point sets, and the structures
they induce \cite{s-iig-53, ww-sg-93, c-t-10, sw-csg-10}. As the name
suggests, for many of the questions one states, an integral is set up whose
solution is the desired quantity, and one remains with the (usually
painful) task of solving the integral\footnote{Historically, the field was not named integral geometry \emph{because} it involved integrals in the calculus sense. The origin of the word, which derives from the German "Integralgeometrie", was coined and popularized by Blaschke in their book. We thank an anonymous reviewer for mentioning this.}. In this paper, we focus mainly on direct combinatorial arguments of said results.

\myparagraph{Closest pair and spread} %
The \emphi{spread} of a point set $\P \subset \Re^d$ is the ratio
between the diameter and the closest pair distance of $\P$. Formally,
it is the quantity
\begin{math}
    \spread = \spreadX{\P} = \diamX{\P} /\cpX{\P},
\end{math}
where
\begin{math}
    \diamX{\P} = \max_{\p,\q \in \P} \dY{\p}{\q}
\end{math}
and
\begin{math}
    \cpX{\P} = \min_{\p,\q \in \P: \p \neq \q} \dY{\p}{\q}.
\end{math}
For a set $\P$ of $n$ points sampled uniformly at random from $\Hd$,
It is not hard to verify \cite{hj-spl-20} that
$\Ex{\cpX{\P}} = \Omega(1/n^{2/d})$. This intuitively suggests that
$\Ex{\spreadX{\P}} = O(n^{2/d})$ - (a formal proof of this requires a
bit more effort).

\myparagraph{Convex-hull}

The Convex-hull of $n$ points in $\Re^d$ has combinatorial complexity
$\Theta(n^{\floor{d/2}})$ in the worst case (here, combinatorial complexity refers to the number of vertices and faces). It can be computed in
$O(n\log n + n^{\floor{ d/2} })$ time \cite{c-ochaa-93}. Surprisingly,
the expected complexity of the convex-hull of random points picked
from $\Hd$ is $O( \log^{d-1} n)$ \cite{bkst-anmsv-78}.  The exact
bound depends on the underlying domain from which the points are sampled. For
example, if the sample is taken from a ball in $\Re^d$, the expected
complexity is $O(n^{(d-1)/(d+1)})$ \cite{r-slcdn-70}. See
\cite{h-ecrch-11} and references therein for more details.  Dwyer
\cite{d-chrpp-88} provides an expected linear time algorithm for
computing the convex hull of a set of points picked from $\Hd$. As
hinted to earlier, the analysis is not elementary and uses heavy tools
to show the result.

\myparagraph{Delaunay triangulation} %
The Delaunay triangulation $\DT$ of $n$ points in $\Re^d$ has
combinatorial complexity $\Theta(n^{\ceil{d/2}})$ in the worst
case. It can be computed in $O(n\log n + n^{\ceil{ d/2} })$ time
\cite{c-ochaa-93}. Dwyer \cite{d-hvdle-91} show that when the points
are uniformly sampled from a $d$-dimensional unit ball (instead of a
$d$-cube), the complexity of the Voronoi diagram (and consequently its
dual, $\DT$) is also linear, and gave an $O(n)$ time expected time
algorithm for constructing it. However, Dwyer's algorithm is involved
and its analysis is nontrivial with reliance on algebraic and integral
tools.

\myparagraph{Minimum spanning trees} %
There is a lot of work on \MST and \EMST (Euclidean minimum spanning
tree). Since \EMST is a subgraph of the \emphw{Gabriel graph} of $\P$
-- that is, the graph where two points $\p,\q \in \P$ are connected by
an edge, if their diametrical ball does not contain any point of $\P$
in its interior.  The Gabriel graph is a subgraph (of the
$1$-skeleton) of $\DTX{\P}$, the Delaunay Triangulation of $\P$. Thus, one can calculate $\DTX{\P}$ (in
linear time), and then run Karger \etal expected linear time \MST
algorithm \cite{kkt-rlafm-95} on $\DTX{\P}$.  The algorithm of Karger
\etal uses as a black box a procedure to identify all the edges in the
graph that are too heavy to belong to a minimum spanning tree, given a
candidate spanning tree. Such spanning tree ``verifiers'' are
relatively complicated to implement in linear time
\cite{h-eslta-09}. Developing deterministic linear time \MST algorithm
is still an open problem, although Chazelle presented
\cite{c-mstai-00} a $O(n + m \alpha(n,m))$ time algorithm where $n,m$ are the number of vertices and edges respectively (as
$\alpha(n,m)$ is at most $4$ for all practical purposes, this is
essentially a linear time algorithm). More bizarrely, a deterministic
\emph{optimal} algorithm is known \cite{pr-omsta-02}, but its running
time complexity is not known.  None of these algorithms can be
described as simple.

For minor-closed graphs, \Mares \cite{m-tltam-04} gave two linear time
algorithms to construct the \MST in $O(n+m)$ time. In the plane, the
Delaunay Triangulation is a planar graph, and thus given the Delaunay
triangulation the \MST can be computed in linear time (this is no
longer applicable, already in 3d).

\myparagraph{Distance selection} %
Given a set $\P$ of $n$ points in the plane, and a number $k$, the
distance selection problem asks for the $k$\th small distance in the
$\binom{n}{2}$ pairwise distances induced by the points of $\P$.  In
the plane, this can be computed in $O(n \log n + k)$ time
\cite{c-esd-01}, or alternatively in $O(n^{4/3} )$ time
\cite{cz-hplsfc-21} for general sets of points. An $(1\pm \eps)$-approximation can be computed in
linear time \cite{hr-nplta-15}.

\subsection{Our results}

We provide simple and elementary proofs for several of the results
mentioned above, and we also provide (conceptually) simple algorithms
for several of the problems mentioned above:
\begin{compactenumA}
    \smallskip%
    \item \textsc{$k$\th{} distance concentration.}  Fix a value of
    $r \in [0,1]$.  Let $f_r = f_r(\P)$ denote the number of pairs
    $\p_i \p_j$ with $\dTY{\p_i}{\p_j}\leq r$, where
    $\dTY{\p_i}{\p_j}$ is the \emphw{torus topology distance} between
    $\p_i$ and $\p_j$ (defined in \Eqref{toroidal_dist_defn}). Note
    that $f_r(\P)\in \{0, ..., \binom{n}{2}\}$. It is not hard to show
    that
    $\Prob{\cardin{f_r - \Ex{f_r}} >\tldOmg(n^{3/2})} \leq
    {1}/{n^{O(1)}}$ using Chernoff's inequality and the union bound,
    where $\tldOmg$ and $\tldO$ hide polylogarithmic terms in
    $n$. However, in \secref{concentrate:k:th:distance}, we show a
    significantly stronger concentration, namely that the interval has
    length $\tldO(n)$ with high probability\footnote{Here, an event
       $A_n$ happens with \emphw{high probability} if
       $\Prob{A_n}\geq 1-{1}/{n^{O(1)}}$.}:
    $\Prob{\cardin{f_r - \Ex{f_r}} >\tldOmg(n)} \leq {1}/{n^{O(1)}}$.
    The new concentration proof uses martingales together with bounded
    differences concentration inequality that can handle low
    probability failure.  To the best of our knowledge this result is
    new, and is an interesting property of random points. (We
    conjectured this claim after observing this behavior, of strong
    concentration, in computer simulations we performed.). The proof is an interesting application of a McDiarmid's inequality variant that allows a (small) probability of large variation, when applying the standard McDiarmid's inequality would otherwise fail.

    \medskip%
    \item \textsc{Convex hull.}
    In \secref{convex:hull}, as a warm-up exercise, we provide an
    $O(n)$ expected time algorithm to construct $\CHX{\P}$, the convex
    hull of $\P$. Dwyer \cite{d-chrpp-88} presented a divide and
    conquer algorithm. Our algorithm is somewhat different as it uses
    a quadtree for the partition scheme, and is the building block for
    the later algorithms.

    \smallskip%
    \item \textsc{Linear complexity of Delaunay triangulation.}
    We provide a new proof that the expected complexity of the
    Delaunay triangulation of $\P$ is linear, where $\P$ is a set of
    $n$ points picked uniformly and independently from $\Hd$. The new
    proof, presented in \secref{del:triangulation}, is simpler and
    more direct than existing proofs. The linear bound is quite easy
    to derive for points in the inner part of the cube (we refer to
    this part of the cube as the \emphw{fortress}), but the outer part
    (i.e., the \emphw{moat}) requires more work because of boundary
    issues.

    \smallskip%
    \item \textsc{Linear time algorithm for Delaunay triangulation.}
    In \secref{d:t:build}, we present an expected linear time
    algorithm for computing the Delaunay triangulation.  The algorithm
    computes, for each point, the points it might interact with, and
    the local Delaunay triangulation of these points. The algorithm
    then stitch these local structures together to get the global
    triangulation.

    \smallskip%
    \item \textsc{Euclidean MST.}
    Since the \MST of $\P$ is a subgraph of (the $1$-skeleton) of
    $\DTX{\P}$, the (general but more complicated) expected linear
    time \MST algorithm from \cite{kkt-rlafm-95} could be applied to
    $\DTX{\P}$ to calculate the \EMST of $\P$ in linear time. For
    $d=2$, it is known that \Boruvka's algorithm implemented
    efficiently\footnote{Some textbook implementations would run in
       $O(n\log n)$ time, even if the graph is planar.} takes linear
    time, since planarity is preserved between rounds.  In particular,
    we conjecture that \Boruvka's algorithm takes linear time when run
    on $\DTX{\P}$, in higher dimensions, but we were unable to prove
    it.

    Instead, in \secref{m_s_t}, we present an algorithm for
    constructing the \EMST of $\P$, in expected linear time, using a
    simple algorithm that is the adaption of \Boruvka's algorithm to
    use divide and conquer over a quadtree storing the points. The
    correct propagation of subtrees of the \MST that can be computed
    when restricted to a subproblem, together with a ``minimal'' set
    of edges that might participate in the \MST, is the main new idea
    of our new algorithm. We believe the new algorithm should be of
    interest when trying to compute \MST{}s, or similar structures,
    for huge graphs where one has to distribute the computation across
    several computers/nodes.

    \smallskip%
    \item \textsc{Distance selection.} %
    We show a simple algorithm for distance selection for $\P$ that
    works in expected $O(n^{4/3} \log^{2/3} n)$ time. The new
    algorithm achieves this running time by partitioning the problem
    into (roughly) $O(n^{2/3})$ special instances involving (roughly)
    $O(n^{1/3})$ points concentrated in ``tiny'' disks, and a set of
    points that lies in a ring, of radius $r$, containing (roughly)
    $O(n^{2/3})$ points. Each of these instances can be solved by a
    direct point-location algorithm in (roughly) $O(n^{2/3})$ time. In
    the general case, one has to rely on a more complicated divide and
    conquer strategy (implemented using cuttings3), together with
    duality, to reach such unbalanced instances that can be solved
    using brute force (see \cite{cz-hplsfc-21} and references
    therein). Thus, the new algorithm provides a new elegant and
    intuitive explanation where the mysterious $n^{4/3}$ term rises
    from, in addition for providing a simple algorithm that might work
    better in practice than previous algorithms.

\end{compactenumA}

\paragraph{A comment on the paper organization.}

Since this paper has many results, and is long, we ordered our results
in such a way, that (hopefully) the first ten pages convey our basic
approach and ideas. We did move some (more minor) proofs to an
appendix.

\section{Preliminaries}

\myparagraph{Notations} %
The $O$ notation hides constants that depend (usually exponentially)
on $d$.

\subsection{\VC dimension and the \TPDF{$\eps$}{eps}-net and
   \TPDF{$\eps$}{eps}-sample theorems}

The main ingredient in almost all our results is the
$\eps$-net/sample theorems. In this subsection, we give a quick
introduction, see \cite{AS-TPM-00} or \cite{h-gaa-11} for more
details.  We do not assume prior knowledge of this topic.

\begin{definition}
    A \emphi{range space} $S = (\GSet, \Family)$ is a pair, where
    $\GSet$ is a set, and $\Family$ is a family of subsets of
    $\GSet$. The elements of $\GSet$ are \emphw{points} and the
    elements of $\Family$ are \emphw{ranges}.
\end{definition}

A subset $B\subseteq \GSet$, is \emphi{shattered} by $\Family$ if the
$\cardin{ \Set{r\cap B}{r \in \Family}} = 2^{\cardin{\GSet}}$. The
\emphi{Vapnik-Chervonenkis} dimension (or $\VC$-dimension) of the
range space $S=(\GSet,\Family)$ is the maximum cardinality of a
shattered subset of $\GSet$.

\begin{example}
    Suppose $\GSet=\Re^2$ and $\Family$ is the set of disks in
    $\Re^2$. For any set of three (not colinear) points
    $T=\{\p_1, \p_2, \p_3\} \subseteq \GSet$, and any subset
    $T'\subseteq T$, one can find a disk containing $T'$, and avoiding
    the points of $T \setminus T'$. Thus, the $\VC$ dimension of disks
    in the plane is $3$. It is easy to verify that no four points can
    be shattered, and thus the $\VC$ dimension of this range space is
    $3$.
\end{example}

\begin{example}
    In general, for points in $\Re^d$ and balls or halfspace ranges,
    the $\VC$ dimension is $d+1$. Another noteworthy range is
    axis-parallel rectangles which have $\VC$ dimension $2d$.
\end{example}

For simplicity of exposition, assume $\GSet$ to be a finite set of
points.  An \emphi{$\eps$-net} captures all ``heavy'' ranges.  That
is, if we sample a ``sufficiently'' large subset $N\subseteq \GSet$,
then any range $r\in \Family$ containing ``enough points'' from
$\GSet$ must also contain a point from $N$ with high probability. The
\emphi{$\eps$-sample} is similar, asserting that for any range
$r\in \Family$, the fractions $\frac{\cardin{N\cap r}}{\cardin{N}}$
and $\frac{\cardin{\GSet\cap r}}{\cardin{\GSet}}$ are $\eps$-close,
with high probability. The formal definition is stated below.
\begin{definition}
    Let $(\GSet,\Family)$ be a range space, and let
    $\FGSet \subset \GSet$ be a finite subset. For $0<\eps <1$, a
    subset $\Net \subseteq \FGSet$, is an \emphi{$\eps$-net} for
    $\FGSet$ if for any range $r\in \Family$, we have
    \begin{math}
        \displaystyle%
        \cardin{r\cap \FGSet}\geq \eps \cardin{\FGSet}
        \implies r\cap \Net \neq \emptyset.
    \end{math}

\end{definition}

\begin{definition}
    Let $(\GSet,\Family)$ be a range space, and let $\FGSet$ be a
    finite subset of $\GSet$. For $\eps \in (0,1)$, a subset
    $\Net \subseteq \FGSet$, is an \emphi{$\eps$-sample} for $\FGSet$
    if for any range $r\in \Family$, we have
    \begin{equation*}
        \cardin{ \frac{\cardin{\Net \cap r}}{\cardin{\Net}}
           - \frac{\cardin{\FGSet \cap r}}{\cardin{\FGSet}} }
        \leq \eps.
    \end{equation*}
\end{definition}
Finally, the $\eps$-net and $\eps$-sample theorems characterizes
quantitatively the size of the sample needed to have the desired
property.
\begin{theorem}[$\eps$-net theorem, \cite{hw-ensrq-87}]
    \thmlab{eps:net}%
    Let $(\GSet,\Family)$ be a range space of \VC-dimension $d$, let
    $\FGSet \subseteq \GSet$ be a finite subset, and suppose
    $\eps>0, \delta<1$. Let $\Net$ be a random sample from $\FGSet$
    with $m$ independent draws, where
    \begin{math}
        \displaystyle%
        m \geq \max\Bigl(\frac{4}{\eps}\log \frac{2}{\delta},
        \frac{8d}{\eps}\log \frac{8d}{\eps} \Bigr).
    \end{math}
    Then $\Net$ is an $\eps$-net for $\FGSet$ with probability at
    least $1-\delta$.
\end{theorem}

\begin{theorem}[$\eps$-sample theorem, \cite{vc-ucrfe-71,vc-ucfoe-13}]
    \thmlab{eps:sample}%
    Let $(\GSet,\Family)$ be a range space, where its \VC-dimension is
    $d$.  Let $\FGSet \subseteq \GSet$ be a finite subset, and suppose
    $\eps>0, \delta<1$ are parameters. Let $\Net$ be a random sample
    of size $m$ from $\FGSet$, where
    \begin{math}
        m \geq \min\Bigl( \cardin{\FGSet}, \frac{32}{\eps^2}\bigl(
        d\log \frac{d}{\eps} + \log \frac{1}{\delta}\bigr) \Bigr).
    \end{math}
    Then, $\Net$ is an $\eps$-sample for $\FGSet$, with probability at
    least $1-\delta$.
\end{theorem}

\newcommand{\vlV}{\frac{c_d\ln n}{n}}%
\newcommand{\vlrV}{\sqrt[d]{\vlV}}%
\subsection{Bounding the moments}
In the following, $n$ is fixed, and let $\P$ be a set of $n$ points
picked randomly, uniformly and independently from
$[0,1]^d$. Throughout, we use the following fixed quantities:
\begin{equation}
    \vl = \frac{c_d\ln n}{n}
    \qquad\text{ and }\qquad%
    \vlr = \vl^{1/d},
    \eqlab{delta:val}
\end{equation}
where $c_d>0$ a sufficiently large constant that depend only on $d$.

Throughout the paper, we often need to bound the moments of the number
of points of $\P$ that lie in some measurable set $\Xi$. The following
technical lemma bounds the expected number of such points.
\begin{lemma}%
    \lemlab{moments}%
    Let $\Xi \subseteq [0,1]^d$ be a measurable set. If
    $\alpha = \volX{\Xi} \geq 1/n$, then for $t > 2e$, we have
    \begin{equation*}
        \Prob{ |\P \cap \Xi| > t\cdot \alpha n } \leq 1/2^{t\alpha n}.
    \end{equation*}
    Furthermore, for any constant $\cPower \geq 1$, we have that
    \begin{math}
        \Ex{\bigl.|\P \cap \Xi|^\cPower} = O\bigl( ( \alpha
        n)^\cPower )
    \end{math}
    (the $O$ hides here a constant that depend on $\cPower$).
\end{lemma}

\begin{proof:e:n}{{\Xlemref{moments}}}{proof:moments}
    The number of points of $\P$ falling into $\Xi$, is a binomial
    distribution, and we have
    \begin{equation*}
        \Prob{\bigl. |\P \cap \Xi| > t\cdot \alpha n }
        =%
        \SoCGVer{\!\!\!}%
        \sum_{i=t \alpha n +1}^n \binom{n}{i} \alpha^i
        (1-\alpha)^{n-i}
        \leq%
        \SoCGVer{\!\!\!}%
        \sum_{i=t \alpha n +1}^n \pth{\frac{n \alpha e}{i}}^i
        \leq
        \SoCGVer{\!\!\!}%
        \sum_{i=t \alpha n +1}^n \pth{\frac{n \alpha e}{2 e \alpha n}}^i
        \leq
        \frac{1}{2^{t\alpha n}},
    \end{equation*}
    since $\binom{n}{i} \leq \pth{\frac{ne}{i}}^i$.  Thus, we
    have %
    \begin{equation*}
        \Ex{\bigl.|\P \cap \Xi|^\cPower}%
        \leq%
        \sum_{t=0}^\infty \bigl((t+1) \alpha n\bigr)^\cPower
        \Prob{\bigl. |\P \cap \Xi| > t\cdot \alpha n }%
        \leq%
        (2 \alpha n)^\cPower \sum_{t=0}^\infty \ffrac{
           t^\cPower}{2^{t\alpha n}} %
        =%
        O\bigl( ( \alpha n)^\cPower )\Bigr.,
    \end{equation*}
    since
    \begin{math}
        \sum_{t=0}^\infty \ffrac{ t^\cPower}{2^{t\alpha n}} %
        \leq \sum_{t=0}^\infty \ffrac{ t^\cPower}{2^{t}} %
        =%
        O(1).
    \end{math}
\end{proof:e:n}

\subsection{Vicinities}
For two points $\p, \q \in \Re^d$, let $\bbY{\p}{\q}$ denote the axis
parallel bounding box of $\p$ and $\q$.  The \emphi{vicinity} of a
point $\p \in [0,1]^d$ is
\begin{math}
    \vicX{\p} =%
    \Set{ \q \in [0,1]^d}{\volX{\bbY{\p}{\q}} \leq \vl},
\end{math}
where $\vl$ is specified in \Eqref{delta:val} (see
\figref{vicinity}).  For a number $x > 0$, let
\begin{math}
    \powX{x} = 2^{\ceil{\log_2 x}},
\end{math}
and observe that $x \leq \powX{x} \leq 2x$, and $\powX{x}$ is a
power of two. This definition is used in the proof of the following claim.

The following claim can be proved using integration -- we provide
an alternative combinatorial proof for the sake of completeness.

\begin{figure}
    \centerline{\includegraphics%
       {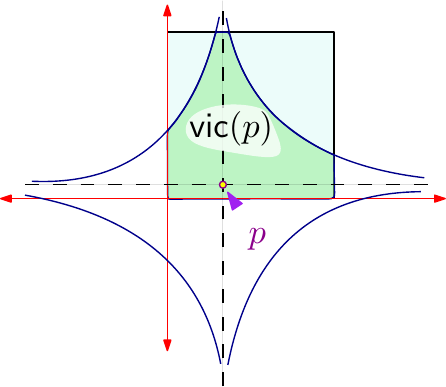}}
    \captionof{figure}{The green region is $[0,1]^2 \cap \vicX{\p}$. The light blue region is $[0,1]^2 \setminus \vicX{\p}$.}
    \figlab{vicinity}
\end{figure}

\begin{lemma}[\cite{chr-puvpp-16}]
    \lemlab{vicinity:vol}%
    For any $\p \in [0,1]^d$, we have
    $\volX{\vicX{\p}} =O( \frac{\log^{d} n}{n} )$.
\end{lemma}
\begin{proof:e}{{\Xlemref{vicinity:vol}}}{v:vol:proof:x}
    Let $\origin$ denote the origin.  The area of a single quadrant of
    the vicinity is maximized when $\p=\origin$. There are $2^d$
    quadrants so we have that
    $\volX{\vicX{\p}} \leq 2^d \volX{\vicX{\origin}}$. To bound the
    later quantity, let $\tau >0$ be an integer such that
    $2^{d}\vl \leq 2^{-\tau} \leq 2^{d+1}\vl$. As such, we have
    \begin{math}
        \tau \leq \lg(1/ \vl) - d = O( \log n).
    \end{math}

    A \emphi{canonical box}, is a box of the form
    $B= \prod_{i=1}^d [0,\alpha_i] $ such that
    $\volX{B} = \prod_i \alpha_i = 2^{-\tau}$, where $\alpha_i$ is a
    power of two, for all $i$.  For any point
    $\q = (\q_1,\ldots, \q_d) \in [0,1]^d$, let
    $\pvX{\q} = \prod_{i=1}^d \q_i $ be the \emphi{point
       volume}. Consider all the points of
    $\q = (\q_1, \ldots, \q_d) \in \vicX{\origin}$ (i.e., these are
    points with $\pvX{\q} \leq \vl$), and let
    $\powX{\q} = ( \powX{\q_1}, \ldots, \powX{\q_d})$.  Observe that
    $\pvX{\powX{\q}} \leq 2^d \vl$. In particular, there exists a
    canonical box that contains $\q$.

    Consider a side $[0,\alpha_i]$ of a canonical box. The number
    $\alpha_i$ is a power of $2$, and
    $1 \geq \alpha_i \geq 2^{-\tau}$. That is,
    $\alpha_i \in \{1,1/2,1/4, \ldots, 2^{-\tau} \}$. Namely, there
    are at most $1+\tau$ choices for the value of each coordinate of a
    canonical box.  As such, the number of canonical boxes is
    $(1+\tau)^{d-1}$, as fixing $d-1$ coordinates forces the value of
    the last coordinate. The volume of a canonical box is
    $\leq 2^{d+1} \vl$. We conclude $\vicX{\origin}$ is covered by the
    union of these boxes, and as such,
    $\volX{\vicX{\origin}} \leq \tau^{d-1} \vl$, which implies the
    claim.
\end{proof:e}
The intuition for vicinities is that for a lot of the problems
discussed in the introduction, any point $\p \in \P$ only needs to
locally consider other points in its vicinity when making decisions of
building the desired structures (i.e. points outside the vicinity of
$\p$ are not relevant for $\p$)

\newcommand{\UCT}{[0,1]^d_T}%

\section{Sharp concentration of the \TPDF{$k$}{k}\th %
   pairwise distance}
\seclab{concentrate:k:th:distance}

Let $\P = \brc{ \p_1, \ldots, \p_n}$ be a random sequence of points,
where $\p_i$ is picked uniformly and independently from $\Hd$.  For
two numbers $x,y \in [0,1]$ their \emphw{toroidal distance} is
$\cardinT{x-y} = \min\bigl( \cardin{x-y}, 1-\cardin{x-y}\bigr)$.  Let
$r$ be a fixed value in $[0,1]$. Let $f_r(\P)$ denote the number of
pairs $\p_i \p_j$ in $\P$ with $\dTY{\p_i}{\p_j}\leq r$. Formally, we have
\begin{equation}
\eqlab{toroidal_dist_defn}
    \dTY{\p_i}{\p_j} = \sqrt{\sum\nolimits_{\ell=1}^d
       \cardin{p_{i}[\ell]-p_{j}[\ell]\bigr.}_T^{2}}
    \qquad\text{and}\qquad%
    f_r(\P) = \cardin{ \Set{ \p_i \p_j }{ i <j \text{ and } \dTY{\p_i}{\p_j}
          \leq r}}.
\end{equation}
is the \emphi{toroidal distance} between $\p_i$ and $\p_j$, and
$\p[\ell]$ denotes the $\ell$\th coordinate of a point $\p \in \Re^d$.
We denote the space $[0,1]^d$ under this toroidal topology by
$\UCT$. Intuitively, this is the space where we allow ``wrap-around''
in $[0,1]^d$, and the shortest distance between two points can be the
wrap-around distance. Using this distance allows us to ignore
artifacts that are generated by the boundary of the hypercube
$[0,1]^d$.

The claim is that the value of $f_r$, which is a number in
$\{0,\ldots, \binom{n}{2}\}$, is strongly concentrated.  Namely, the
interval of integers containing $f_r$, with high probability, is
``short''.  Showing a bound of $\tldO(n^{3/2})$ on the number of
values in this interval is doable via Chernoff's inequality and using
the union bound. The resulting guarantee is of the form
\begin{math}
    \Prob{\cardin{f_r - \Ex{f_r}} >\tldOmg(n^{3/2})} \leq
    {1}/{n^{O(1)}}.
\end{math}
Here, we show a significantly stronger concentration with the interval
containing $\tldO(n)$.

We conjecture this result is true for the Euclidean distance, but
handling the boundary cases proved to be quite challenging.  Hence the
simplifying Toroidal topology assumption. We observed this strong
concentration, for both the Toroidal and Euclidean case, in computer
simulations.

Consider the closed ball
\begin{math}
    \ballYT{\p}{r}=\Set{x\in \Hd }{ \dTY{\p}{x}\leq r}
\end{math}
in $\UCT$.  We next bound the \VC dimension of such balls (as a side,
Gillibert \etal \cite{glm-vabt-22} bounded the \VC-dimension of
axis-parallel boxes in this space by $O(d \log d)$).
\begin{lemma}\RefProofInAppendix{v:c:t}
    \lemlab{vc-toroidalball}%
    For $\mathcal{A}=\Set{\ballYT{p}{r}}{p\in \Re^d }$, the \VC
    dimension of the range space $(\Hd, \mathcal{A})$ is $O(1)$.
\end{lemma}
\begin{proof:e}{{\Xlemref{vc-toroidalball}}}{v:c:t}
    A toroidal ball consists of at most $O(2^{2d})$ regions $R_i$,
    each region being the intersection of a ball and at most $2d$ half
    spaces (corresponding to the boundaries of $\Hd$). The \VC
    dimension of balls and halfspaces is $d+1$, so the \VC dimension
    of their intersection (and hence each region) is $O(1)$. Taking
    the union of the at most $O(2^{2d})$ regions, implies the \VC
    dimension is at most $O(1)$, via standard argumentation
    \cite{h-gaa-11}.
\end{proof:e}

Next, we would like to apply Chernoff-like style inequalities to bound
the probability of deviation from the expectation. The most relevant
inequality here is McDiarmid's inequality for bounded differences of
martingales. Unfortunately, one cannot use McDiarmid's inequality
directly because by ``sliding'' a ball $\ball$ in $[0,1]^d$, the
number of points inside $\ball$ might change by $O(n)$. Of course for
random point set $\P$, this is highly unlikely (the change is more
likely to be $O(\sqrt{n})$) and so we will have to use a variation of
McDiarmid's inequality that allows a ``bad'' event, where the
difference might be large but happening with a small probability, and
a ``good'' typical event where the difference is bounded.

Consider the following extension of McDiarmid’s inequality (for
bounded differences of martingales) where the differences are only
bounded with high probability \cite{k-emidb-02}. It will be useful to
view $\P$ from two different views\footnote{As with most things in
   life.}, one as a set of individual points, and the second as a
product $\Omega=\prod_{1\leq i \leq dn}\Omega_i$ of $dn$ probability
spaces for each coordinate.
\begin{definition}[\cite{k-emidb-02}]
    Let $\Omega_1, ..., \Omega_{m}$ be probability spaces.  Let
    $\Omega = \prod_i \Omega_i $, and let $X$ be a random variable on
    $\Omega$. The variable $X$ is \emphw{strongly difference-bounded}
    by $(b,c,\hdelta)$ if the following holds. There is a ``bad''
    subset $B\subseteq \Omega$, where $\hdelta=\Prob{w\in B}$. In
    addition, we require that
    \begin{compactenumi}
        \smallskip
        \item If $\omega, \omega' \in \Omega$ differ only in the
        $k$\th coordinate, and $\omega \not \in B$ then
        $\cardin{X(\omega)-X(\omega')} \leq c$.

        \smallskip
        \item Furthermore, for any $\omega, \omega'\in \Omega$
        differing only in the $k$\th coordinate,
        $\cardin{X(\omega)-X(\omega')}\leq b$.
    \end{compactenumi}
\end{definition}

To decipher this definition consider the case that $c < b$: the
difference between ``bad'' pairs can be large, but the difference
between ``good'' (or mixed) pairs is small. The quantity $b$ behaves
like the ``worst'' case difference, $c$ is the ``typical'' difference,
and $\hdelta$ is the probability of the bad event happening.

\begin{lemma}[\cite{k-emidb-02}, Corollary~3.4]
    \lemlab{kutin-mcdiarmid}%
    Let $\Omega_1, ..., \Omega_m$ be probability spaces.  Let
    $\Omega = \prod_{1\leq i \leq m} \Omega_i$ and let $X$ be a random
    variable on $\Omega$ which is strongly difference-bounded by
    $(b, c, \hdelta)$. Let $\mu = \Ex{X}$. Then, for any
    $\tau > 0$, and any $\alpha > 0$, we have
    \begin{math}
        \Prob{\bigl.\cardin{X-\mu} \geq \tau} \leq%
        2\Bigl[\exp{\bigl(-\frac{\tau^2}{2m(c+b\alpha)^2}\bigr)}+
        \frac{m}{\alpha}\hdelta \Bigr].
    \end{math}
\end{lemma}

In the following, let $\P + \p = \P \cup \{\p\}$ and
$\P - \p = \P \setminus \{\p\}$.

\begin{lemma}%
    \lemlab{strongly-bounded-fr}%
    The random variable $f_r$ is strongly difference-bounded by
    \begin{equation*}
        (b,c,\hdelta) := (n-1, O(\sqrt{n \log n }), {1}/{n^{O(1)}}).
    \end{equation*}
\end{lemma}
\begin{proof:e:n}{{\Xlemref{strongly-bounded-fr}}}{s:b:f}
    If one moves only one point of $\P$, at most $n-1$ pairwise
    distances involved with this point can change, implying that
    $b \leq n-1$.

    By the \thmrefY{eps:sample}{$\eps$-sample theorem}, and
    \lemref{vc-toroidalball}, a sample of size
    $O\bigl( \eps^{-2} \log n \bigr)$ is an $\eps$-sample for Toroidal
    balls, with high probability. Interpreting $\P$ as an
    $\eps$-sample for $\Hd$, implies that this holds for $\P$ with
    $\eps = \sqrt{\vl} = \sqrt{\frac{c_d \ln n}{n}}$ for sufficiently
    small constant $c_d > 0$. Let
    $\vSX{d}=\volX{\ballYT{\q}{r}\bigr.}$, for any point
    $\q \in [0,1]^d$.  The number of points in distance $\leq r$ from
    a point $\p \in [0,1]^d$, is
    $X_\p = \cardin{\P \cap \ballYT{\p}{r}}$.  Hence, for \emph{any}
    Toroidal ball we have
    \begin{equation}
        \cardin{X_p - \Ex{X_\p}}
        =%
        \cardin{\bigl.\cardin{\P \cap \ballYT{p}{r}}-\vSX{d}n} \leq \eps n =
        \sqrt{c_d n \log n} = \tldO(\sqrt{n})
        \eqlab{sample}
    \end{equation}
    assuming $\P$ is indeed an $\eps$-sample. Note that the bound above crucially uses the Toroidal distance properties. Furthermore, the set
    $\P - \p$, formed by removing any point $\p \in \P$, is an
    $\eps$-sample, and this holds with high probability for all such
    subsets.

    This readily implies that for any two points
    $\p, \p' \in [0,1]^d$, we have
    \begin{equation*}
        \cardin{X_\p - X_{\p'}}
        \leq
        \cardin{X_\p - \vSX{d} n}
        +
        \cardin{\vSX{d} n - X_{\p'} }
        = O( \sqrt{ n \log n}).
    \end{equation*}

    Picking a point $\p \in \P$, and a point $\p' \in [0,1]^d$, and
    setting $\P' = \P - \p + \p'$, we are interested in bounding the
    ``typical'' difference between $f_r(\P)$ and $f_r(\P')$ (this
    would be the value of $c$). We have
    \begin{align*}
      \cardin{f_r(\P) - f_r(\P')}
      \leq%
      \cardin{X_\p - X_{\p'}} + O(1)
      = O( \sqrt{n \log n}).
    \end{align*}
    This implies that $c = O( \sqrt{n \log n})$. This calculation
    fails, only if $\P$ fails to be an $\eps$-sample, which happens
    with probability $\hdelta \leq 1/n^{O(1)}$.
\end{proof:e:n}

\begin{theorem}%
    \thmlab{concentrated-variance}%
    For constant $c'$ sufficiently large, we have
    \begin{math}
        \Prob{\cardin{f_r - \Ex{f_r}} > c' n \log n } %
        \leq%
        {1}/{n^{O(1)}}.
    \end{math}
\end{theorem}
\begin{proof:e:n}{{\thmref{concentrated-variance}}}{c:v:proof}
    This follows readily by plugging the parameters of
    \lemref{strongly-bounded-fr} into \lemref{kutin-mcdiarmid}. Note
    that in our case, $m=n, b=n-1$, and $c=\sqrt{c_dn\log
       n}$. Choosing $\alpha={1}/{n}$ and $\hdelta \leq {1}/{dn^3}$
    (which can be ensured by making $c_d$ sufficiently large), and
    $\tau=\Omega(n \log n)$, the result follows by straightforward
    calculations from \lemref{kutin-mcdiarmid}.  For example, plugging
    $\tau = 100\sqrt{c_d }n\log n$ into \lemref{kutin-mcdiarmid}
    yields (for sufficiently large $c_d$):
    \begin{math}
        \displaystyle%
        \Prob{\cardin{f_r(\P)-\Ex{f_r(\P)}} > 100\sqrt{c_dd}n\log n}
    \end{math}
    \begin{math}
        \leq%
        2\bigl[\exp{\bigl(-\frac{10000c_d dn^2 \log(n)^2
           }{2dn(2\sqrt{c_d n \log n}+1)^2}\bigr)}+ \frac{dn^2}{dn^3}
        \bigr]%
        \leq%
        \frac{4}{n}.
    \end{math}
\end{proof:e:n}

\section{Warm-up: Computing the convex hull in linear time}
\seclab{convex:hull}

We present here an algorithm for computing the convex hull of $\P$ in
$O(n)$ time. This will serve as a warm-up as the tools here will be
used later on.  \myparagraph{Algorithm}

Given $\P$, we build $T$, a quadtree of height
$h=\lceil (\log_2 n)/{d} \rceil$ and insert the points $\P$ in $O(n)$
time to its leaves -- this can readily be done by storing the points
in the grid formed by the leafs using hashing (or just direct array
indexing).

The algorithm computes the convex hull via a bottom-up traversal of
the tree. It starts by computing the convex hull (potentially empty)
for each leaf of the quadtree using any brute force algorithm. For a
node $v$ at level $k$, the algorithm takes the computed convex-hulls
of its children, extracts all their vertices and stores it in a set $S$,
and computes the combined convex-hull of $S$, using off-the-shelf
algorithm \cite{c-ochaa-93} in $O( |S|\log |S| + |S|^{\floor{d/2}})$
time.

\myparagraph{Analysis} The algorithm correctness is immediate. We next
prove an inferior upper bound on the expected complexity of the random
convex-hull that holds with higher moments.
\begin{lemma}%
    \lemlab{convex-hull-moments}%
    Let $\cardin{\CHX{\P}}$ denote the number of vertices in the
    convex hull of $\P$. For any integer $\cPower > 0$, we have
    $\Ex{\cardin{\CHX{\P}}^\cPower} = O(\log^{O(\cPower d)}n)$ (the
    constant hidden by the $O$ depends on both $\cPower$ and $d$).
\end{lemma}
\begin{proof:e:n}{{\Xlemref{convex-hull-moments}}}{c:h:m:proof}
    Let $\p$ be a vertex of the convex-hull $\CHX{\P}$, and consider a
    tangent (hyper)plane $h$ to $\CHX{\P}$ that passes through
    $\p$. The plane $h$ separates $\CHX{\P}$ from one of the vertices
    of the $\Hd$, say $\q$. Let $R$ be the axis parallel box with $\p$
    and $\q$ as antipodal vertices.

    The \VC dimension of axis aligned boxes is $2d$; see
    \cite{h-gaa-11}. By the $\eps$-net theorem, a sample of size
    $O( d \eps^{-1} \log n)$ is an $\eps$-net for axis aligned boxes,
    with probability $\geq 1- 1/n^{O(d)}$. Setting
    $\eps = \vl = c_d (\log n) / n$, it follows that $R$ contains a
    point of $\P$ with high probability. It follows that all the
    vertices of $\CHX{\P}$ are in the vicinity of some vertex of
    $\Hd$. Let $\Xi$ be the union of the vicinities of the vertices of
    $\Hd$. By \lemref{vicinity:vol}, we have that
    $\alpha = \volX{\Xi} = O( (\log n)^d / n)$. Applying
    \lemref{moments} to $|\Xi \cap \P|^\cPower$ now implies the claim.
\end{proof:e:n}

\begin{lemma}%
    \lemlab{runtime:convex:hull:quad}%
    The above algorithm computes $\CHX{\P}$ is $O(n)$ expected time.
\end{lemma}
\begin{proof:e:n}{{\Xlemref{runtime:convex:hull:quad}}}{r:c:q}
    Consider the root of the quadtree -- it has $2^d$ children, and
    let $\P_i$ be the set of points of $\P$ stored in the $i$\th
    child. Let $n_i = |\P_i|$ and $m_i = |\VX{\CHX{\P_i}}|$. We have
    that $\sum_i n_i = n$, and $\Ex{n_i} = n/2^d$. In particular,
    using Chernoff's inequality we have that $n_ i \leq (7/8)n$ with
    high probability. Similarly, we have that
    $\Ex{m_i^\cPower} = O(\log^{O(\cPower d)}n)$ by
    \lemref{convex-hull-moments}.  Let $m = \sum_i m_i$, and observe
    that computing the convex-hull at the top most level takes
    $O( m\log m + m^{\floor{d/2}}) = O(m^d) = O(2^d \sum_{i} m_i^d)$.
    Thus, ignoring the construction time of the quadtree itself, we
    have the recurrence
    \begin{equation*}
        T(n) = O\pth{
           \Ex{ \smash{
                 \textstyle \sum\nolimits_i m_i^d
              } \bigr. }
        }
        + \sum_i T(n_i)
        =%
        \log^{O(d^2)}n
        + \sum_i T(n_i),
    \end{equation*}
    and the solution to this recurrence is $O(n)$.
\end{proof:e:n}

\section{Complexity of the Delaunay triangulation %
   of random points}
\seclab{del:triangulation}

Here we show that the Delaunay triangulation of $\P$ has linear
complexity in expectation.

\myparagraph{Background on Delaunay triangulations} %
A simplicial complex $\DT$ over a set $\P$ is a set system with the
(hyper) edges being subsets of $\P$, such that for any
$\sigma, \sigma' \in \DT$, we have that $\sigma \cap \sigma' \in \DT$.
An edge of $\DT$ is a \emphi{simplex}. A simplex is \emphi{$k$
   dimensional} if the affine space its points span is $k$
dimensional.  Simplices of dimension $0,1$ and $2$ are
\emphi{vertices}, \emphi{edges} and (two dimensional) \emphi{faces},
respectively.

For a point $\p \in \Re^d$, and a radius $r > 0$, let $\ballY{\p}{r}$
denote the open \emphi{ball} of radius $r$ centered at $\p$.  For any
points $\p_1, \ldots, \p_k \in \Re^d$, let
$\pencilX{ \p_1, \ldots, \p_k }$ denote the \emphi{pencil} of
$\p_1, \ldots, \p_k$: the set of all \emph{open} balls $\ball$ in
$\Re^d$, such that their boundary sphere passes through
$\p_1, \ldots, \p_k$. If $k=d+1$, and the points are in general
position, the pencil is a single ball $\circumX{\p_1, \ldots, \p_k}$
bounded by the \emphi{circumscribed sphere} of these points.

The \emphi{Delaunay triangulation} $\DT=\DT(\P)$ of $\P$ is a
simplicial complex, where $\simplex \in \DT$ $\IFF$ there is a ball
$\ball \in \pencilX{ \simplex }$ such that $\ball \cap \P =\emptyset$.
The Delaunay triangulation has the property that if the set of points
$\P$ is a random set then the points are in general position with
probability $1$ (i.e., \emphi{almost surely}), and it is then uniquely
defined.

\begin{figure}[t]
    \begin{minipage}{0.4\linewidth}
        \centering%
        \includegraphics[scale=0.65]{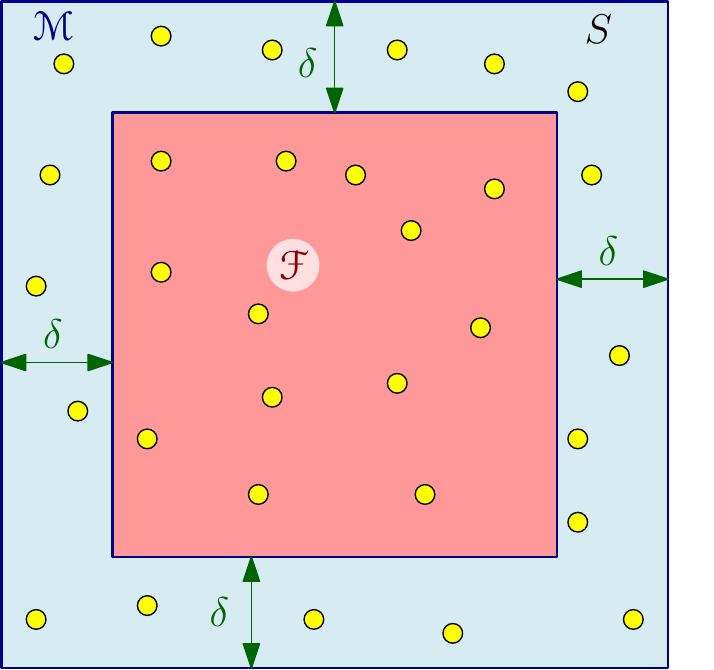}%
        \captionof{figure}{The \emphi{inner fortress} $\IR$ in
           red. The \emphi{moat} $\Moat$ in light blue. Here
           $\vlr=\vlrV$.}
        \figlab{critical-strip}
    \end{minipage}
    \hfill
    \begin{minipage}{0.4\linewidth}
        \centering \includegraphics{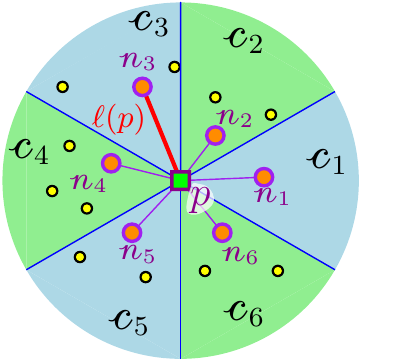}%
        \captionof{figure}{Definitions of
           $\nnX{i} = \nnY{\cone_i}{\p}$ and $\reachX{\p}$.}
        \figlab{six-cones-and-definitions}
    \end{minipage}
\end{figure}

\subsection{A linear bound in the interior of the hypercube }
\seclab{d:t:fortres}

Let $\IR = [\vlr, 1-\vlr]^d$ be the \emphi{fortress} of $\Hd$, and
$\Moat = \Hd \setminus \IR$ be its \emphi{moat}, where $\vlr = \vlrV$.
See \Eqref{delta:val} and \figref{critical-strip}.

\begin{defn}
    Consider a ray emanating from a point $\q$ in a direction $v$ in
    $\Re^d$.  A \emphi{cone} of angle $\alpha$ is the set of all
    points $\p \in \Re^d$, such that the angle between $\p - \q$ and
    $v$ is at most $\alpha$. The point $\q$ is the \emphi{apex} of the
    cone.
\end{defn}

One can cover space around a point with $O_d(1)$ cones, with angle
$\pi/12$, to cover all of $\Re^d$

\begin{lemma}[\cite{dgl-ptpr-96}]
    \lemlab{filling-space-rd-cones}%
    For any point $\p \in \Re^d$, one can construct a set
    $\ConesX{\p}$ of $2^{O(d)}$ cones with apex $\pa$ and angle
    $\pi/12$, such that $\cup \ConesX{\p} =\Re^d$.
\end{lemma}

The following shows that all these cones are not empty, if the apex
$\p$ is in the fortress.

\begin{lemma}%
    \lemlab{one-point-in-each-cone}%
    For any $\p \in \IR$, and any cone $\cone \in \ConesX{\p}$, we
    have $\Prob{ \cone \cap (\P -\p) = \emptyset} < 1/n^{O(1)}$.
\end{lemma}
\begin{proof:e:n}{{\Xlemref{one-point-in-each-cone}}}{o:p:e:c:proof}
    Since $\p\in \IR$, it is at a distance of at least $\vlr$ from the
    boundary of $[0,1]^d$.  Thus, $\alpha = \volX{\cone \cap [0,1]^d}$
    $= \Omega(\vlr^d) = \Omega(\vl) = \Omega \bigl( (\log n) / n
    \bigr)$, by \Eqref{delta:val}. This implies that the probability
    that $\cone$ does not contain any of the points of $\P - \p$ is at
    most
    \begin{equation*}
        (1-\alpha)^{n-1}%
        \leq%
        \exp \pth{\bigl. -\alpha (n-1) }
        \leq
        \exp \pth{ - c \ln n }
        =
        \frac{1}{n^{O(1)}},
    \end{equation*}
    where $c$ is a constant that can be made to be arbitrarily large
    by increasing the value of $c_d$. The result now follows by
    applying the union bound to all the cones in $\ConesX{\p}$.
\end{proof:e:n}

\begin{defn}
    \deflab{reach}%
    For a point $\p \in \P \cap \IR$ and a cone
    $\cone \in \ConesX{\p}$, let $\nnY{\cone}{\p}$ denote the nearest
    neighbor to $\p$ in $\cone \cap (\P - \p)$. The \emphi{reach} of
    $\p$ is
    \begin{math}
        \reachX{\p} = \max_{\cone \in \ConesX{\p}}
        \dY{\p}{\nnY{\cone}{\p}},
    \end{math}
    see \figref{six-cones-and-definitions}.
\end{defn}

The following bounds (in expectation) $\ell(\p)$ and the distance
between $\p$ and $\nnY{\cone}{\p}$ for $\cone\in \ConesX{\p}$.

\begin{lemma}\RefProofInAppendix{n:n:cone:decay:proof}
    \lemlab{n:n:cone:decay}%
    For any point $\p \in \P \cap \IR$, a cone $\cone \in \Cones(\p)$
    and $t\leq \vlr$, we have:
    \begin{compactenumI}
        \item
        $\Prob{\dY{\p}{\nnY{\cone}{\p}}>t} \leq \exp(- c\, t^dn )$,
        where $c$ is constant that depends only on the dimension.

        \item
        $\Prob{\reachX{\p}>t}\leq f(t,n) = \exp\bigl(O(d) -c \,t^d n
        \bigr)$,

        \item $\Ex{\reachX{\p}} = \Theta(1/\sqrt[d]{n})$, and

        \item the reach of all the points of $\P \cap \IR$ is bounded
        by $\vlr$, with probability $1/n^{O(d)}$.
    \end{compactenumI}
\end{lemma}

\begin{proof:e}{{\Xlemref{n:n:cone:decay}}}{n:n:cone:decay:proof}
    (I) If $\dY{\p}{\nnY{\cone}{\p}}>t$ then the set
    $R = \cone \cap \ballY{\p}{t}$ contains no points of $\P$.  The
    volume of $\cone \cap \ballY{\p}{t}$, for $t\leq \vlr$, is
    $\alpha = \Omega(t^d)$. In this case, all the points of $\P - \p$
    must avoid $R$. This happens with probability at most
    $(1-\alpha)^{n-1} \leq \exp\bigl(-\alpha(n-1) \bigr) \leq
    \exp\bigl(- c t^dn\bigr)$.

    (II) By the union bound, and \lemref{filling-space-rd-cones}, we
    have
    \begin{math}
        \Bigl.  \Prob{\reachX{\p}>t}%
        \leq%
        \sum_{\cone \in \Cones(\p)} \Prob{\nnY{\cone}{\p}>t}%
        \leq%
        f(t,n).%
    \end{math}

    (III) Observe that $\ballY{\p}{1/2\sqrt[d]{n}}$ contains no other
    points of $\P - \p$ with constant probability. This implies that
    $\Ex{\reachX{\p}} = \Omega( 1/ \sqrt[d]{n})$. The upper bound
    $\Ex{\reachX{\p}} = O( 1/ \sqrt[d]{n})$ follows by the above
    exponential decay, as a straightforward calculation shows.

    (IV) Setting $t = \vlr$, and using the union bound implies this
    part.
\end{proof:e}

\begin{definition}
    \deflab{ward}%
    For $\p\in \P\cap \IR$, the \emphi{influence} of $\p$ is
    \begin{math}
        \wardX{\p}=\Set{\q\in \P }{\dY{\p}{\q}\leq 2\reachX{\p}}
    \end{math}
\end{definition}

\begin{figure}[h]
    \vspace{-0.6cm}%
    \centering%
    \phantom{}\hfill%
    {\includegraphics{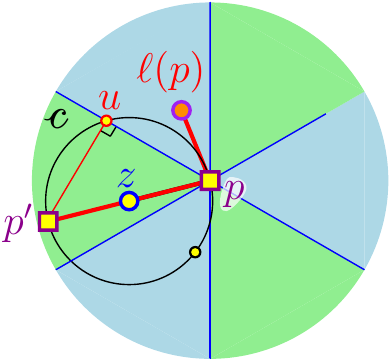}}%
    \hfill%
    \includegraphics[page=2]{figs/cone_containment}%
    \hfill\phantom{}%
    \caption{Sketch of proof of \lemref{internal-2m-proof} }
    \figlab{proof-2m-sketch}
\end{figure}

Importantly, all the Delaunay edges adjacent to a point $\p \in \IR$
are contained in $\p$'s region of influence.  Hence, locally, it is
sufficient to only consider points in the influence when computing the
Delaunay triangulation.
\begin{lemma}%
    \lemlab{internal-2m-proof}%
    For any point $\p \in \P\cap \IR$, If $\p\q\in \DT$, then
    $\q\in \wardX{\p}$.
\end{lemma}
\begin{proof:e:n}{{\Xlemref{internal-2m-proof}}}{i:2:p}
    Consider the largest (open) ball $\ball$ with $\p$ on its
    boundary, that does not contain any point of $\P$ in its interior,
    and let $r$ be its radius and $\cen$ be its center, see
    \figref{proof-2m-sketch}. We claim that $r \leq \reachX{\p}$,
    which would imply that $\dY{\p}{\q} \leq 2r \leq
    2\reachX{\p}$. Assume that $r > \reachX{\P}$, and consider any
    cone $\cone \in \ConesX{\p}$, such that $\cen \in \cone$. Let
    $\p'$ be the diametrical point on $\partial \ball$ to
    $\p$. Consider any point $\pu \in \cone \setminus \ball$. The
    distance $\p \pu$ is minimized if $\pu \in \partial \ball$, but
    then $\angle \p\pu\p'$ forms the right angle of a right triangle. Observe that
    $\angle \pu \p \p'< 30^\circ$ since the cone angle is at most
    $30^\circ$.  But then
    \begin{equation*}
        \dY{\p}{\pu}
        =%
        \dY{\p}{\p'} \cos \angle \p' \p \pu
        =%
        2r \cos \angle \p' \p \pu
        >%
        2\reachX{\p} \cos 30^\circ
        =%
        (2 \sqrt{3} /2) \reachX{\p} >
        \reachX{\p}.
    \end{equation*}
    However, $\ball \cap (\P - \p) =\emptyset$ implies that the
    closest point in $(\P- \p) \cap \cone$ to $\p$ has distance larger
    than $\reachX{\p}$, which contradicts the definition of
    $\reachX{\p}$.
\end{proof:e:n}

We next bound the moments of the size of the set of points inside the
influence of a point.
\begin{lemma}\RefProofInAppendix{i:n:o:o}
    \lemlab{internal:neighbors:f}%
    For $\p\in \P\cap \IR$, and any constant $\cPower \geq 1$, we have
    $\Ex{ |\wardX{\p}|^\cPower \bigr. }=O_\cPower(1)$, see
    \defref{ward}.
\end{lemma}
\begin{proof:e}{{\Xlemref{internal:neighbors:f}}}{i:n:o:o}
    Let $L = \cardin{\wardX{\p}}$.  We break $\P$ into two roughly
    equal sets $\P_1$ and $\P_2$ (this is done before sampling the
    locations of the points). Let $\reachC_i = \reachY{\p}{ \P_i}$ be
    the reach of $\p$ in $\P_i$ for $i\in\{1,2\}$. Arguing as above,
    as $\p \in \IR$, this quantity is well defined. Let $U_i$ be the
    number of points of $\P_{3-i}$ in the ball
    $\ball_i = \ballY{\p}{\reachC_i}$. Clearly,
    $\Ex{L} \leq \Ex{U_1} + \Ex{U_2}$, as $\ball_i$ contains more
    points of $\P_{3-i}$ than the ball defined by the reach of the
    whole set.

    So, let $\psi$ be the minimum value such that
    $f(\psi,n/2) \leq 1/2$, where $f$ is the function defined in
    \lemref{n:n:cone:decay}. It is easy to verify that
    $\psi = O (1/n^{1/d})$. Let $\ball_0 = \ballY{\p}{\psi}$, and,
    for $i>0$, let
    $\ring_i = \ballY{\p}{i\psi} \setminus \ballY{\p}{(i-1)\psi}$.
    Observe that $\volX{\ball_0} = O(1/n)$, and
    $\volX{\ring_i} = O(i^{d}/n)$.  By \lemref{moments}, we have that
    $N_0 = \Ex{\bigl. |\P_2 \cap \ball_0|^\cPower } = O(1)$, and
    $N_i = \Ex{ |\P_2 \cap \ring_i|^\cPower }= O(i^{d\cPower})$.  We
    have that $\Ex{L^\cPower \bigr.} = O(T)$, where
    \begin{equation*}
        T
        =%
        \sum_{i=0}^\infty N_i \Prob{ \ell > i\psi }%
        \leq
        \sum_{i=0} O_\cPower(i^{\cPower d}  f(i \psi, n/2))
        =%
        \sum_{i=0} O_\cPower(i^{\cPower d} / 2^i)
        =
        O_\cPower(1),
    \end{equation*}
    since
    \begin{math}
        f(i \psi, n/2) =%
        \exp\bigl(O(d) -c \,i^d \psi^d n \bigr) \leq
        \pth{\exp\bigl[O(d) -c \, \psi^d n \bigr]}^{i^d} \leq
        1/2^{i^d}.
    \end{math}
\end{proof:e}

Combining everything, we get the main result for points in the fortress.
\begin{lemma}%
    \lemlab{internals-are-fine}%
    Let $\DT$ be the Delaunay triangulation of $\P$.  The expected
    number of simplices that include any point of $\P\cap \IR$ is
    $O(n)$.
\end{lemma}
\begin{proof:e:n}{{\Xlemref{internals-are-fine}}}{i:a:f}
    Let $\tau=|\wardX{\p}|$.  All the vertices of a simplex of the
    Delaunay triangulation containing $\p$ must have all its vertices
    in $\wardX{\p}$ by \lemref{internal-2m-proof}. Thus, the number of
    such simplices, of all dimensions, is bounded by
    $\sum_{i=0}^{d+1} \binom{\tau}{i} = O( \tau^d)$.  By
    \lemref{internal:neighbors:f}, we have $\Ex{O( \tau^d)}=O(1)$.
\end{proof:e:n}

\subsection{The complexity of the Delaunay triangulation near the
   boundary}

We are now left with the tedious technicality of handling points that
are too ``close'' to the boundary\footnote{Ha, the boundary! A source
   of unmitigated delight to the authors, and hopefully also to the
   readers.}. The idea is to use a similar argumentation to the above,
but to replace the influence ball induced by the reach by a
different region. This inflated region contains significantly more
points, but since the number of points in the moat is small, this
would still be linear overall.  We remind the reader that the
\emph{moat} is the area $\Moat = \Hd \setminus [\vlr, 1-\vlr]^d$, see
\figref{critical-strip} and \Eqref{delta:val}.

The following is an immediate consequence of \lemref{moments} and
\lemref{vicinity:vol}.

\begin{lemma}
    \lemlab{degree}%
    For any point $\p \in \P$, we have
    $X = \cardin{\vicX{\p} \cap \P} = O( \log^d n)$ with probability
    $\geq 1 -1/n^{O(1)}$.
\end{lemma}

\begin{lemma}\RefProofInAppendix{sub:simplex}
    \lemlab{sub:simplex}%
    Consider an axis parallel box $B = \prod_{i=1}^d [\p_i,\q_i]$, and
    assume that there is a ball $\ball$ that contains the points
    $\p = (\p_1,\ldots, \p_d)$ and $\q = (\q_1,\ldots, \q_d)$. Then
    $\ball$ contains a $d$-dimensional simplex $\simplex$ defined by
    $d+1$ vertices of $B$, such that the volume of this simplex is
    $\geq \volX{B}/d!$. More generally, this holds for any ball that
    contains two diametrical vertices of $B$.
\end{lemma}
\begin{proof:e}{\Xlemref{sub:simplex}}{sub:simplex}
    The proof is by induction on $d$. The claim is immediate if $d=1$.

    \begin{figure}[h]
        \centering%
        \includegraphics{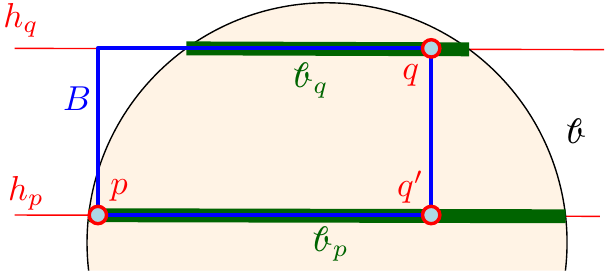}%
        \qquad\begin{minipage}[b]{3cm}
            \captionof{figure}{}
            \figlab{b:b}

            \bigskip \bigskip

        \end{minipage}
    \end{figure}

    For $d>1$, the idea it to provide a path along the edges of the
    box $B$ between the two vertices that is contained in $\ball$ --
    the convex-hull of this path would provide the desired
    simplex. So, consider the two hyperplanes $h_\p \equiv x_d = \p_d$
    and $h_\q \equiv x_d = \q_d$, see \figref{b:b}. Consider the two
    balls $\ball_\p = \ball \cap h_\p$ and
    $\ball_\q = \ball \cap h_\q$. Both balls have the same center if
    we ignore the $d$\th coordinate, and one of them must have a
    bigger (or equal) radius to the other. Assume that $\ball_\p$ has
    the bigger radius, and observe that it as such must contain the
    point $\q' = (\q_1, \ldots, \q_{d-1}, \p_d)$. This implies that
    the segment $\q \q' \subseteq \ball$. By induction there is a path
    on the edges of $B \cap h_\p$ between $\p$ and $\q'$, which
    implies that there is a path on the edges of $B$ between $\p$ and
    $\q$ that lies inside $\ball$.
\end{proof:e}

For points $\p\in \Moat$, as the following testifies, one needs to
consider only simplices and points that are in $\vicX{\p}$. This is
indeed a larger region than the influence region used before, but is
small enough for our purposes.
\begin{lemma}\RefProofInAppendix{moat:degree}
    \lemlab{moat:degree}%
    Let $\p$ be any point in $\P$. With high probability, there are at
    most $O( \log^{d^2} n)$ simplices in $\DT = \DTX{\P}$ that
    contains $\p$ as a vertex.  Furthermore, all the points
    neighboring $\p$ in $\DT$ must be in $\vicX{\p}$, with probability
    $\geq 1-1/n^{O(d)}$.
\end{lemma}

\begin{proof:e}{\Xlemref{moat:degree}}{moat:degree}
    The \VC dimension of simplices in $\Re^d$ is $O( d^2 \log d)$ as
    it is the intersection of $d+1$ halfspaces, each of \VC dimension
    $d+1$, see \cite{h-gaa-11}. By the $\eps$-net theorem, a sample of
    size $O\bigl( (d^2 \log d) \eps^{-1} \log n \bigr)$ is an
    $\eps$-net for simplices, with high probability. Interpreting $\P$
    as an $\eps$-sample for $\Hd$, implies that this holds for $\P$
    with $\eps = \vl/d!$, where $\vl = (c_d\ln n)/n$, see
    \Eqref{delta:val} (by making $c_d$ sufficiently large).

    Consider a point $\q \in \P$, such that $\q \notin \vicX{\p}$, and
    assume that $\p \q \in \DTX{\P}$. This implies that there is a
    close ball $\ball$ that has $\p$ and $\q$ on its boundary, and no
    other points of $\P$ in its interior. By \lemref{sub:simplex},
    there is an (open) simplex $\simplex$ of volume
    $\geq \volX{\bbY{\p}{\q}} /d! \geq \eps$ that contains $\p$ and
    $\q$ on its boundary, and it is contained inside $\Hd \cap
    \ball$. But since $\P$ is an $\eps$-net for simplices, it follows
    that there is a point of $\P$ in $\simplex$, which is a
    contradiction.

    We conclude that all the edges adjacent to $\p$ in $\DT$ must be
    to points in $\vicX{\p}$. But there are at most
    $t = O( \log^{d} n)$ such points, by \lemref{degree}. Since any
    simplex involving $\p$ in $\DT$ must use only points that are in
    the vicinity, it follows that the number of simplices (of all
    dimensions) adjacent to $\p$ in $\DT$ is bounded by
    $\sum_{i=0}^{d} \binom{t}{i} = O(t^d)$.
\end{proof:e}

Finally, we show that the complexity of the Delaunay triangulation in
the moat is sublinear.
\begin{lemma}
    \lemlab{moat:d:t}%
    Let $\DT$ be the Delaunay triangulation of $\P$.  The expected
    number of simplices in $\DT$ that include any point of
    $\P\cap \Moat$ is $o(n)$.
\end{lemma}
\begin{proof}
    We have
    $\alpha = \volX{\Moat} \leq 2d \vlr = O( \sqrt[d]{(\log n) /n })$,
    see \Eqref{delta:val}. Thus, the expected number of points of $\P$
    in $\Moat$ is $\alpha n = O( n^{1-1/d} \log n)$. (As usual, this
    bound holds with high probability.) By \lemref{moat:degree}, the
    total number of simplices in the Delaunay triangulation of $\P$
    involving points in the moat is bounded by
    $O( \alpha n \log^{O(d^2)} n ) = o(n)$, with high probability.
\end{proof}

\myparagraph{The result}

Combining \lemref{internals-are-fine} and \lemref{moat:d:t} implies the
following.

\begin{theorem}
    For fixed $d$, the complexity of the Delaunay triangulation of a
    set of $n$ random points picked uniformly and independently in
    $[0,1]^d$ is $O(n)$ in expectation.
\end{theorem}

\section{Constructing the Delaunay triangulation in linear %
   time}
\seclab{d:t:build}

\subsection{Algorithm}

We established above that the (expected) complexity of the Delaunay
triangulation is linear by giving a (linear sized) superset of
vertices/simplices that are a superset of the features of $\DT$. We
are now left with the task of extracting the features that do appear
in $\DT$. Recall, that the input is a set $\P$ of $n$ random points
from $\Hd$.

\myparagraph{I: Computing the Delaunay simplices attached %
   to points in $\P \cap \IR$}

Let $N = \ceil{n^{1/d}}$.  The algorithm throws the points of $\P$
into a $N \times \cdots \times N$ uniform grid covering $\Hd$. This
can be done in linear time using hashing, where one can retrieve a
list of all the points stored in a grid cell in constant time. Here a
grid cell is uniquely represented by an integer tuple from
$\{0,1\ldots, N-1\}^d$. Formally, we map a point
$\p =(\p_1, \ldots, \p_d) \in [0,1)^d$, to the grid cell with
\emphi{id} $\idX{\p} = (\floor{\p_1 N}, \ldots, \floor{\p_d N} )$; see
\cite{h-gaa-11}.

For a point $\p \in \P \cap \IR$, the algorithm computes the reach
$\reachX{\p}$ by performing a marching cubes algorithm computing the
intersection of the grid with the ball $\ballY{\p}{r_i}$, where
initially $r_i = 2^i/N$, for $i=0,1,\ldots$. The algorithm uses
scanning to compute the point set $\PA_i = \ballY{\p}{r_i} \cap \P$ by
extracting all the points stored in the intersecting grid cells. The
algorithm stops in the $i$\th iteration, if all cones in $\ConesX{\p}$
contains at least one point of $\P$. At this point one can compute the
reach of $\p$ by computing for each cone $\cone \in \ConesX{\p}$ the
closest point in $\cone \cap \P_i$ to $\p$. The algorithm then
computes the point set $\P_\p = \P \cap \ballY{\p}{2r_i}$, and
computes the Delaunay triangulation of $\P_\p$ using any standard
algorithm for computing Delaunay triangulation. Finally, the algorithm
extract the star of $\p$ from the computed triangulation, and store
it. As a reminder, the \emphi{star} of $\p$, denoted by $\starX{\p}$,
is the set of all the simplices in the triangulation that contains
$\p$. The algorithm repeats this process for all the points of
$\P \cap \IR$, and returns the union of all the stars computed.

\myparagraph{II: Computing the Delaunay simplices attached %
   to points in $\P \cap \Moat$}

The algorithm builds an orthogonal range searching data structure on
the points $\P \cap \Moat$ (and not on all the whole point set $P$).
Next, for each $\p \in \P\cap \Moat$, the algorithm constructs the set
of $O(\log^{d-1} n)$ canonical boxes $\CBoxes_\p$ (as defined in the
proof of \lemref{vicinity:vol}) that their union covers
$\vicX{\p}$. Then for each $\bx\in \CBoxes_\p$, it queries the data
structure for points set $\P_\bx =\bx\cap \Moat \cap \P$. Next, it
loops over $\q \in \P_\bx$ and adds points in $\P_\bx \cap \vicX{\p}$
to the computed set $N_{\Moat}(\p) = \P \cap \Moat \cap \vicX{\p}$.
Next, using the above grid, it computes the set
$N_{\IR}(\p) = \P \cap \ballY{\p}{2 \vlr}$.  Finally, the algorithm
computes the Delaunay triangulation of
$\P_\p = N_{\Moat}(p)\cup N_{\IR}(p)$ using a standard algorithm and
extracts the star $\starX{\p}$ of $\p$, from the computed
triangulation, and stores it.  The algorithm repeats this for each
$\p\in \P \cap \Moat$ and returns the union of all stars computed for
all $\p$.

\subsection{Analysis}

In the following, we prove that the output of the algorithm is correct
with probability $\geq 1-1/n^{O(d)}$ and the expected running time is
$O(n)$.

\myparagraph{Part I: The fortress}

The correctness of the algorithm is implied by the following claim.

\begin{lemma}
    \lemlab{shaving-is-good}%
    For $\p \in \P\cap \IR$, we have that $\simplex \in \starX{\p}$
    if and only if $\simplex \in \DTX{\P}$.
\end{lemma}
\begin{proof}
    If $\simplex \in \DTX{\P}$, then by \lemref{internal-2m-proof},
    $\simplex \subseteq \wardX{\p}$. This implies that
    $\simplex \in \DTX{\wardX{\p}}$, which implies that $\simplex$ is
    in the computed set $\starX{\p}$.

    If $\simplex \in \starX{\p}$, then the circumball of $\simplex$
    does not contain any point of $\wardX{\p}$ in its interior. If
    this ball contained any point of $\P$ in its interior, then it
    must be further than $2\reachX{\p}$, but this is not possible by
    the argument used in the proof of \lemref{internal-2m-proof}.
\end{proof}

\begin{lemma}
    The above algorithm runs in expected $O(n)$ time.
\end{lemma}
\begin{proof}
    The Delaunay triangulation of $n$ points in $\Re^d$, can be
    computed in $O(n^{\ceil{d/2}} + n \log n) = O(n^d)$
    \cite{m-cgitr-94}. As such, we have that the expected running time
    is
    \begin{math}
        \Ex{ \sum_{ \p \in \P \cap \IR} O(|\wardX{\p}|^d )} = O( n),
    \end{math}
    by \lemref{internal:neighbors:f}.
\end{proof}

\myparagraph{II: The moat} Let $\DTX{\P}_{\Moat}$ denote the set of
simplices in $\DTX{\P}$ with some vertices in $\Moat$. The correctness
of the algorithm is implied by the following.
\begin{lemma}
    \lemlab{sub:moatsimplices}
    For all $\p \in \P \cap \Moat$, we have
    $\simplex \in \starX{\p}$ if and only if $\simplex \in \DTX{\P}_{\Moat}$ with
    probability $\geq 1-1/n^{O(d)}$.
\end{lemma}
\begin{proof}
    Consider a simple $\simplex \in \DTX{\P}_{\Moat}$ with
    $\p \in \VX{\simplex}$.  If $\simplex$ contains a point
    $\q \in \VX{\simplex}$ that is in the fortress $\IR$, and is
    outside $\ballY{\p}{2\vlr}$, then $\reachX{\q} >\vlr$, and
    \lemref{n:n:cone:decay} implies that this happens with probability
    $<1/n^{O(d)}$.  Thus, we have that
    $V(\simplex)\subseteq \P\cap \bigl((\Moat \cap \vicX{p}) \cup
    \ballY{\p}{2\delta}\bigr) \subseteq \P_\p$.  Thus, the empty ball
    $\ball$ in $\DTX{\P}_{\Moat}$ that circumscribes $\simplex$ is
    still empty in $\P_\p$, $\VX{\simplex} \subseteq \P_\p$, and thus
    $\simplex \in \starX{p}$.

    If $\simplex \in \starX{p}$, then there is an empty ball $\ball$
    that circumscribes $\simplex$ and is a witness to this. Assume for
    the sake of contradiction that $\ball$ is not empty, and let $\q$
    be the closest point to $\p$ in $\ball \cap( \P \setminus \P_\p)$.
    If $\q \in \Moat$, then $\q \notin \vicX{\p}$ (as
    $\P \cap \Moat \cap \vicX{\p} \subseteq \P_\p$). The probability
    for that this happens is $< 1/n^{O(d)}$ by \lemref{moat:degree}.
    If $\q \in \IR$, then the cone $\cone \in \ConesX{\q}$ that
    contains $\p$, can not contain any closer point to $\q$ (than
    $\p$) from $\P$. Namely, the reach of $\q$ is bigger than $2\vlr$,
    and probability for that is $\leq 1/n^{O(d)}$, by
    \lemref{n:n:cone:decay}.
\end{proof}

\begin{lemma}
    \lemlab{runtime:delaunay:linear:moat}%
    The above algorithm runs in expected $O(n)$ time.
\end{lemma}
\begin{proof}
    We have
    $\nMoat = \Ex{ |\P \cap \Moat|} \leq n \cdot d \cdot \vlr =
    O(n^{1-1/d} \log n)$.  Building the orthogonal range searching
    data-structure of $\P \cap \Moat$ takes
    $O(n + \nMoat \log^d n ) = O( n)$.

    For any $\p \in \P \cap \Moat$, computing
    $\ballY{\p}{2\vlr}\cap \P$ (using the grid) takes $O( \log n)$
    time (and this bound holds with high probability).  Computing the
    points in the vicinity of $\p$ in the moat takes
    $O( \log^{O(d)} n )$ time -- indeed each orthogonal range query
    takes $O( \log^{d} n)$ time, and there are $O( \log^d n)$ such
    queries. Finally, the time to compute the Delaunay triangulations
    $\P_\p$, is $O( \log^{O(d^2)} n )$.

    Putting everything together, we have that the expected running
    time of the second part of the algorithm is
    $O( n + \nMoat \log^{O(d^2)} n ) = O(n)$.
\end{proof}

\begin{theorem}
    For fixed $d$, and a uniformly and independently sampled point set
    $\P \subseteq[0,1]^d$ of size $n$, the above algorithm computes
    the Delaunay triangulation $\DTX{\P}$ of $\P$ in expected $O(n)$
    time. The algorithm succeeds with high probability.
\end{theorem}

\section{Constructing the \MST in linear time}
\seclab{m_s_t}%

\subsection{Preliminaries}

\begin{lemma}
    \lemlab{mst:short:edges}%
    Let $\MSTP$ be the \MST of $\P$. The longest edge in $\MSTP$ has
    length $\leq \vlr =\sqrt[d]{ c_d (\log n)/n}$ (see
    \Eqref{delta:val}), with probability $\geq 1- 1/n^{O(d)}$, where
    $c_d$ is a sufficiently large constant.
\end{lemma}
\begin{proof}
    Let $\p \q$ be the longest edge in $\MSTP$. Observe that
    diametrical ball $\ball$ defined by $\p$ and $\q$ can not contain
    any points of $\P$ in its interior, as such a point $\pz$, would
    induce a cycle $\p \pz \q$ with $\p \q$ being the longest edge,
    which implies that it is not the \MST. The volume of $\ball$ is
    minimized if its center lies in one the corners of $[0,1]^d$. We
    conclude that the region $R = \ball \cap [0,1]^d$ has
    \begin{math}
        \volX{R } = \Omega_d( \dY{\p}{\q}^d/2^d) = \Omega_d(
        \dY{\p}{\q}^d).
    \end{math}
    Furthermore, $R$ is formed by the intersection of a hyperbox with
    ball, and the \VC dimension of such ranges is $O(d)$
    \cite{h-gaa-11}. The point set $\P$ can be interpreted as an
    $\eps$-net for such ranges, with
    $\eps = \vlr^d /2^d = \Omega_d\bigl( (\log n) / n)$, with high
    probability.  We conclude that if $\dY{\p}{\q} \geq \vlr$, then
    $\P$ fails as an $\eps$-net, which implies the claim.
\end{proof}

\begin{defn}
    \deflab{g:c:y}%
    The \emphi{Yao graph} $\GC = \GCX{\P}$ \cite{y-cmstk-82} of $\P$
    formed by connecting two points $\p,\q \in \P$ by an edge if $\q$
    is the nearest point to $\p$ in one of the cones of $\Cones(\p)$
    (see \lemref{filling-space-rd-cones}). Let $\GCY{\P}{\vlr}$ be the
    graph $\GC$ after removing from it all the edges with length
    $\geq \vlr$.
\end{defn}

It is well known that this graph contains the \MST of $\P$
\cite{y-cmstk-82}.

\begin{lemma}
    \lemlab{y:graph:vlr}%
    Let $\P$ be a set of $n$ points picked uniformly at random from
    $[0,1]^d$. One can compute the graph $\GCY{\P}{\vlr}$ in $O(n)$
    expected time.
\end{lemma}
\begin{proof}
    We store the points of $\P$ in a uniform grid with roughly
    $\Theta(n)$ cells in $[0,1]^d$. For every point $\p \in \P$, and
    every cone $\cone \in \ConesX{\p}$, we perform a marching cube
    algorithm to compute the closest point to $\p$ in $\cone \cap
    \P$. If the search distance exceeds $\vlr$, we abort the search.

    For a point $\p$ in the fortress $\IR$, computing the edges around
    $\p$ takes $O(1)$ time in expectation, by
    \lemref{internal:neighbors:f}. For points in the moat, their
    number is $O( n^{1-1/d} \log n)$, with high probability, and the
    search for each point is truncated after the distance exceeds
    $\vlr$. Per point, such a search takes $O( \log n)$ time. It
    follows that the overall expected running time is
    $O( n + n^{1-1/d} \log^2 n) = O(n)$.
\end{proof}

\myparagraph{A refresher on \Boruvka's algorithm}

Let $\G = (\VV,\Edges)$ be an undirected graph with $n$
vertices and $m \geq n$ edges, and weights on the edges. \Boruvka's
algorithm creates an empty forest $F_0$ over the vertices. Let
$\CC_{i-1}$ be the set of connected components of $F_{i-1}$. For
$\p \in \P$, let $\sigma_{i-1}(v)\in \CC_{i-1}$ denote the connected
component of $v$ in $F_{i-1}$.  While $|\CC_{i-1}|\geq 2$, for each
connected component $C \in \CC_{i-1}$, the algorithm adds the cheapest
edge leaving $\VX{C}$ to some other connected component of
$F_{i-1}$. Let $F_i$ be the resulting forest from $F_i$ after adding
these edges. The final forest is the desired \MST.

Each rounds takes time $O(m)$, and for any $i$ we have
$\cardin{\CC_{i}}\leq \cardin{\CC_{i-1}} / 2$. Thus, \Boruvka's
algorithm takes $O(m\log n)$ time.

\subsection{An \TPDF{$O(n \log n)$}{O(n log n)} time algorithm}

The underling graph in our case is
$\G(\P) = \bigl(\P, \Set{uv}{u,v, \in \P}\bigr)$ where the weight of
each edge is the distance between its endpoints. A naive
implementation of \Boruvka on $\G(\P)$ would require roughly quadratic
time.

\begin{lemma}%
    \lemlab{mst:log:n}%
    For a set $\P$ of $n$ random points in $[0,1]^d$, one can compute,
    in $O(n\log n)$ time, the euclidean minimum spanning tree of
    $\G(\P)$.
\end{lemma}

\begin{proof}%
    One can compute the graph $\GCY{\P}{\vlr}$, see \defref{g:c:y}, in
    $O(n)$ expected time, using \lemref{y:graph:vlr}. By
    \lemref{mst:short:edges}, this graph contains the \MST, which can
    be computed in $O(n \log n)$ time using \Boruvka's algorithm.
\end{proof}

We did some experiments on \Boruvka's algorithm, depicted in
\figref{degree_distribution_boruvka}.

\begin{figure}
    \centering%
    \includegraphics[width=0.8\textwidth]{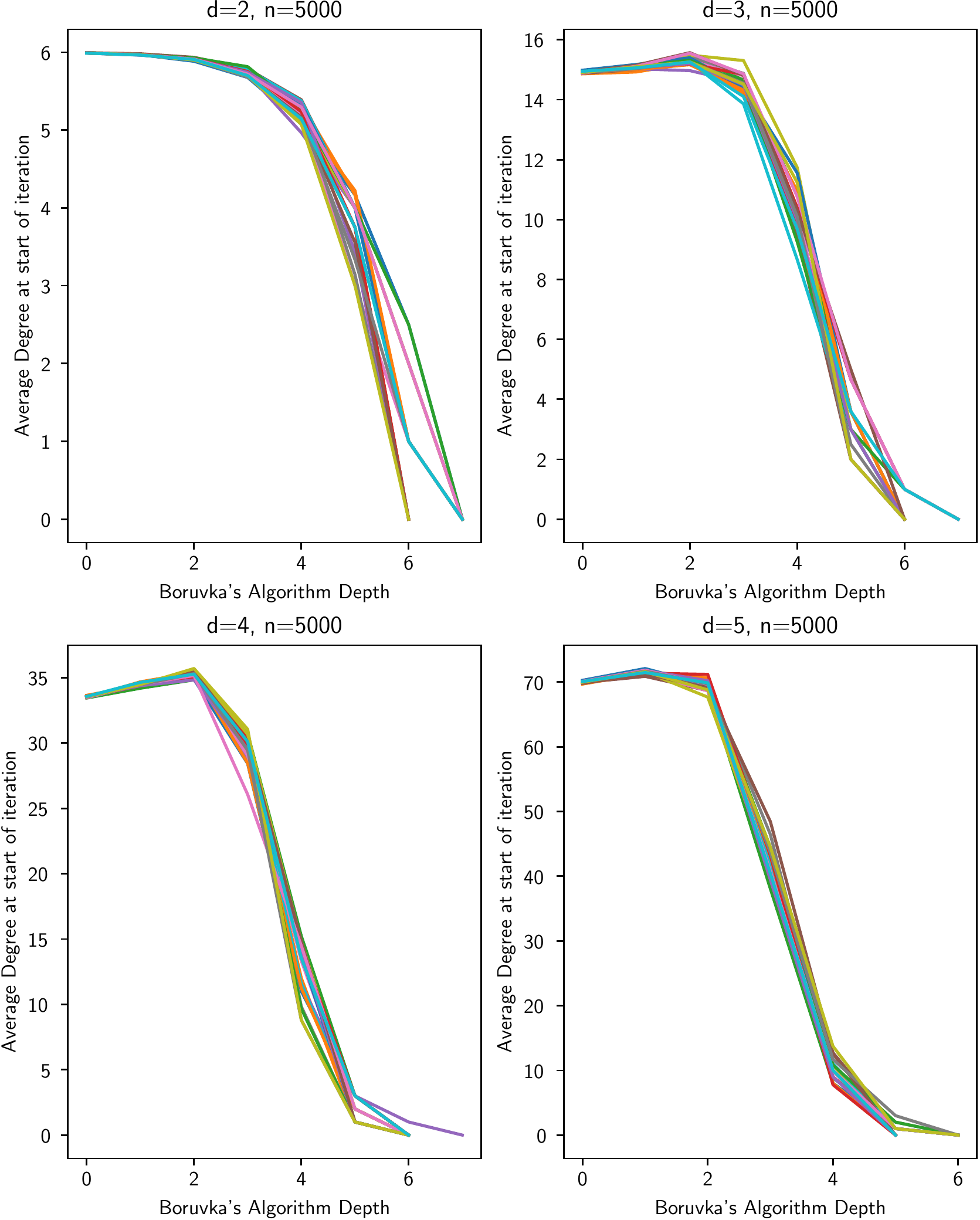}
    \caption{We randomly sample $20$ samples $\P_i, 1\leq i \leq 20$
       where $\P_i \subseteq [0,1]^d$ and $\cardin{\P_i}=5000$. We run
       \Boruvka algorithm with $\DTX{\P_i}$ as input for
       $d\in \{2,3,4,5\}$. In each iteration of \Boruvka algorithm, we
       record the current average degree of the components, and plot
       the average degree progression for the $20$ different samples.}
    \figlab{degree_distribution_boruvka}
\end{figure}

\subsection{Adapting \Boruvka to divide and conquer}

The algorithm precomputes the graph $\GA = \GCY{\P}{\vlr}$. Next, we
turn \Boruvka into a geometric divide and conquer algorithm. To this
end, let $\cell \subseteq [0,1]^d$ be some axis-parallel cube, and
consider computing the \MST of $\P \cap \cell$. Without any outside
information, the output can only be a forest that is part of the final
\MST, and a set of candidate edges that might participate in the final
\MST.  To this end, the algorithm splits $\cell$ into $\nu = 2^d$
identical subcells $\cell_1, \ldots, \cell_\nu$.

The algorithm recursively computes the \MST of
$\P_i = \P \cap \cell_i$, for all $i$. Specifically, the edges of the
\MST are edges of $\GA$, and as such, all the edges of the \MST with
exactly one endpoint in $\P_i$ are in the cut
$\Gamma(\P_i) = \Set{uv \in \EGX{\GA}}{u \in \P_i, v \in \P \setminus
   \P_i}$. Intuitively, the size of this cut is quite small (roughly)
$O(n^{1-1/d})$, and we can identify the vertices in $\P_i$ adjacent to
such edges. These vertices are \emphi{portals}, the set of all portals
in $\P_i$ is denoted by $\portalsX{\P_i}$.

\myparagraph{\Boruvka's algorithm with portals} Imagine running
\Boruvka only on the points of $\P_i$. In every round, each connected
component (in the current spanning forest) chooses the shortest edge
in the cut it defines, and add it to the constructed forest. The catch
is that if a connected component contains a portal point, then it
might be part of a larger tree (in the larger forest) that is outside
$\P_i$. As such, this cut is no longer well defined (as it involves
vertices and edges outside $\P_i$). Thus, a connected component that
contains a portal is \emphi{frozen} -- it can no longer choose edges
to add to the spanning tree.
During a \Boruvka round, all the components that are active (i.e., not
frozen), each chooses the shortest edge in the cut they induce --
note, that an active component might choose an edge connected to a
frozen component. Thus, a frozen component might grow by active
components attaching themselves to it. The algorithm continue doing
rounds till all components are frozen.

A natural implementation of \Boruvka is via collapsing each tree in
the forest being constructed into a single node, and among parallel
edges with the same endpoints, preserving the cheapest edge of the
bunch. Thus, the execution on the modified \Boruvka on $\P_i$ results
in an induced graph $\G_i$ over $\portalsX{\P_i}$ -- where the
surviving edges are potential edges for use by the \MST later on.

\myparagraph{Pruning} The number of edges of $\G_i$ is potentially too
large. The algorithm computes the \MST of $\G_i$ (treating it as its
own graph, ignoring portals) running the standard \Boruvka algorithm
on $\G_i$. The algorithm deletes from $\G_i$ all the edges that do not
appear in the computed \MST.

To recap -- every vertex of $\G_i$ is a collapsed tree forming part of
the final \MST. All the edges of $\G_i$ are candidate edges that might
appear in the final \MST -- all these edges form a spanning tree of
$\G_i$. See \figref{rbs:1} and \figref{rbs:2} for a toy dry run on the
\Boruvka step and the pruning step.

\begin{figure}[t]
    \centering%
    \begin{tabular}{ccc}%
      \includegraphics[width=0.3\linewidth,page=1]%
      {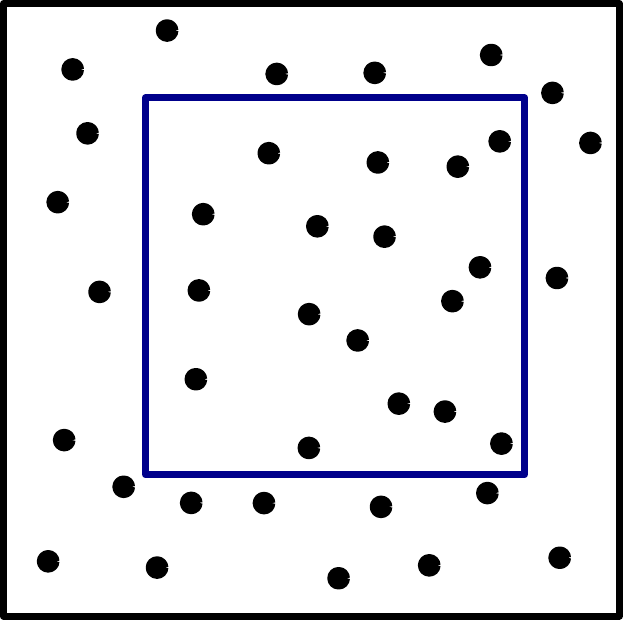}%
      \quad%
      &%
        \includegraphics[width=0.3\linewidth, page=2]%
        {figs/restricted_boruvka_step_1}%
        \quad%
      &%
        \includegraphics[width=0.3\linewidth, page=3]%
        {figs/restricted_boruvka_step_1}%
      \\%
      (a) & (b) & (c)\\[0.2cm]
      \includegraphics[width=0.3\linewidth,page=4]%
      {figs/restricted_boruvka_step_1}%
      \quad%
      &%
        \includegraphics[width=0.3\linewidth, page=5]%
        {figs/restricted_boruvka_step_1}%
        \quad%
      &%
        \includegraphics[width=0.3\linewidth, page=6]%
        {figs/restricted_boruvka_step_1}%
      \\%
      (d) & (e) & (f)
    \end{tabular}
    \caption{(a) shows the axis parallel cube (in blue) and the points
       inside that we restrict ourselves to. $(b)$ shows some of the
       edges of $\GCY{\P}{\vlr}$ inside the cube. $(c)$ shows the
       \emphi{portal} vertices in red, and all other points in
       blue. $(d)$ shows the connected components initially for
       \Boruvka algorithm. $(e)$ shows the edges in cyan that were
       added by restricted \Boruvka to $\EMST$ of $\P$ in the first
       round. (f) Shows the new components after round one of \Boruvka
       algorithm (note the previous blue vertex is now red because it
       joined a component with a portal).  See \figref{rbs:2} for the
       rest.  }
    \figlab{rbs:1}
\end{figure}

\begin{figure}[t]
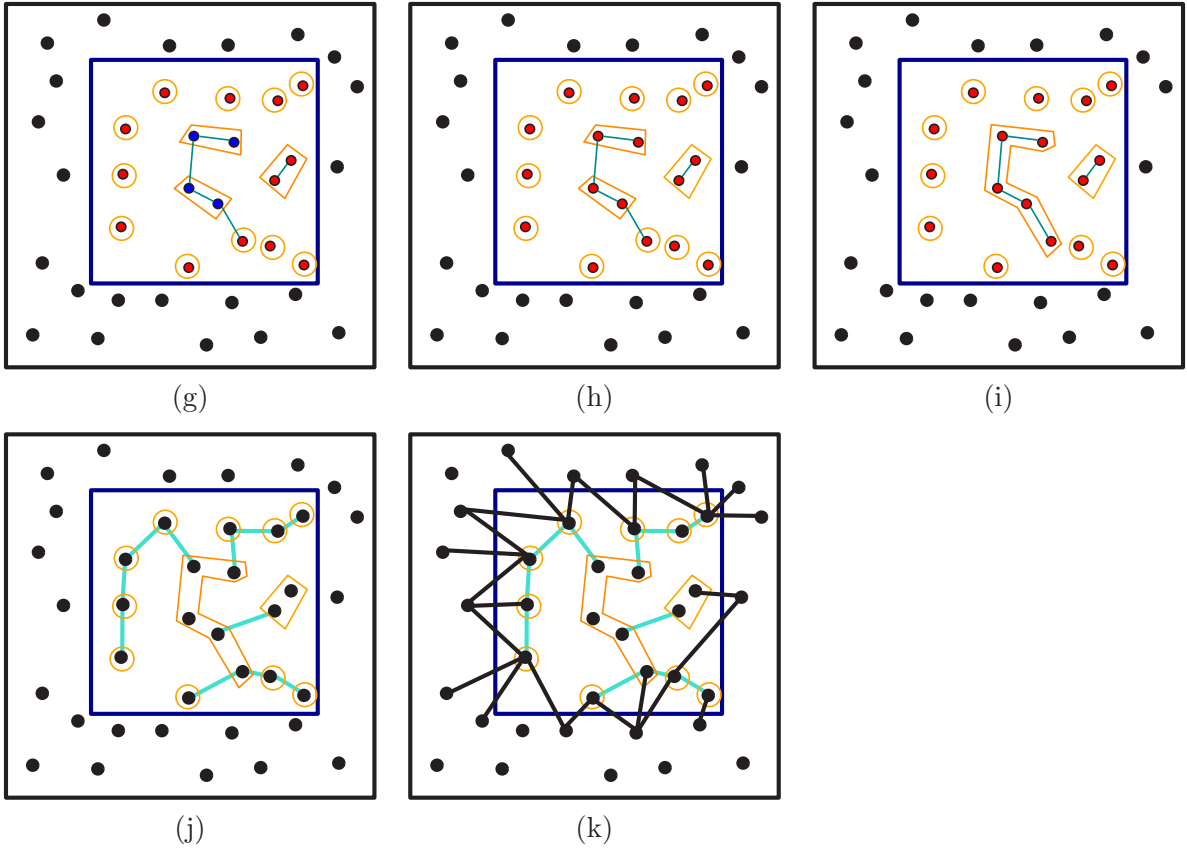

    \begin{tabular}{ccc}
      \includegraphics[width=0.3\linewidth,page=7]%
      {figs/restricted_boruvka_step_1}%
      \quad%
      &%
        \includegraphics[width=0.3\linewidth, page=8]%
        {figs/restricted_boruvka_step_1}%
        \quad%
      &%
        \includegraphics[width=0.3\linewidth, page=9]%
        {figs/restricted_boruvka_step_1}%
      \\%
      (g) & (h) & (i)\\[0.2cm]
      \includegraphics[width=0.3\linewidth,page=10]%
      {figs/restricted_boruvka_step_1}%
      \quad%
      &%
        \includegraphics[width=0.3\linewidth, page=11]%
        {figs/restricted_boruvka_step_1}%
        \quad%
      \\%
      (j) & (k) \\
    \end{tabular}
    \caption{(g) again shows the new edges added to $\EMST$ in the
       second round of \Boruvka. $(i)$ shows the final connected
       components since all components have a portal. $(j)$ shows the
       edges of the minimum spanning tree of the components which
       might be in the \EMST of $\P$. $(k)$ shows the final graph
       returned by the restricted \Boruvka algorithm.  }
    \figlab{rbs:2}
\end{figure}

\myparagraph{The conquer stage}

The algorithm recursively computes the (collapsed) graphs
$\G_1, \ldots, \G_\nu$, for $i=1,\ldots, \nu$.  Next, the algorithm
computes the set of portals $\partial = \portalsX{\cell \cap \P}$,
which is contained in $\cup_i \portalsX{\P_i}$. Let
$\EG_1 = \bigcup_{i < j} (\P_i, \P_j)$ be the set of all possible
edges between the subproblems. Let $\EG_2 = \EG_1 \cap\EGX{
   \GA}$. Next, the algorithm computes the graph
$\G_\cell = \cup_i \G_i \cup \EG_2$. The algorithm runs the modified
\Boruvka with portals, described above, on the graph $\G_\cell$, with
$\partial$ being the set of portals (thus, all the vertices comping
from the children are portals in their own subproblem, but some of
them lose their portal status as they migrate to the parent
subproblem).

\myparagraph{The overall algorithm}

We apply the above algorithm to $[0,1]^d$ and $\P$. Note that the root
has no portals, so the output is a single tree which is the \MST.

\myparagraph{Some low level implementation details} %
We throw the points into a uniform grid over $[0,1]^d$, with each cell
having volume $\Theta(1/n)$. We construct the quadtree over this grid
in the natural way. We register each edge of $\GA$ with the lowest
node of the quadtree that contains both endpoints. This can be done in
$O(1)$ time per edge using a data-structure for \LCA queries in $O(1)$
time. Now, scanning the edges, each vertex can compute the level in
the quadtree where it stops being a portal. The \LCA operation can be
replaced by computing the level of the grid that contains a segment --
using the floor operation and bit operations, this can be done in
$O(1)$ time, see \cite{h-gaa-11}.  The rest of the algorithm
implementation is as described above.

\subsection{Analysis of the new MST algorithm}
\apndlab{m:s:t:analysis}

Clearly, edges that are added to an active component are edges that
are minimal in their respective cuts, and thus must appear in the
final \MST. The more mysterious step in the pruning stage -- let
$\p\q$ be an edge that was deleted by the pruning stage from
$\G_i$. Observe that there is a path $\pi$ between $\p$ and $\q$ in
the graph of $\G_i$ using edges that are shorter than $\p\q$. Namely,
$\p\q$ is the longest edge in a cycle, and can not appear in the final
\MST.

\subsubsection{Running time analysis}

\begin{lemma}
    \lemlab{few:portals}%
    Let $\cell$ be a quadtree cell of depth $i$. Then, the number of
    portals in $\cell \cap \P$ is bounded by
    $O( (n/2^{id} )^{1-1/d} \log^2 n)$, with high probability. This
    also bounds the total number of edges in $\GA$ adjacent to these
    vertices.
\end{lemma}
\begin{proof}
    A point of $\p \in \cell$ that is in distance larger than $\vlr$
    from the boundary of $\cell$ can not be a portal, since $\GA$ does
    not contain such long edges. The volume of the moat $\Moat_\cell$
    containing such points is bounded by the surface area of $\cell$
    multiplied by $\vlr$. That is $\alpha = (2d \cdot /2^{id})
    \vlr$. Each such moat point has with high probability $O( \log n)$
    edges in $\GA$. It follows that the expected number of portal
    edges is $O( \alpha n ) = O( (n/2^{id} )^{1-1/d} \log^2 n)$, as
    long as $\alpha > (\log n)/n$, by \lemref{moments},
\end{proof}

\begin{lemma}
    \lemlab{runtime:mstlinear}
    The above algorithm runs in $O(n)$ expected time.
\end{lemma}
\begin{proof}
    Let $\nu=2^d$, and $\P_1, \ldots, \P_\nu$ be the points sent to
    the children of the root of the quadtree.  Let
    $n_i' = \cardin{\portalsX{\P_i}}$, for all $i$.  By
    \lemref{few:portals}, $n_i = O(n^{1-1/d} \log^2 n )$ with high
    probability, and this also bounds the number of edges these
    portals have. Note, that each $\G_i$ has exactly $n_i'-1$
    edges. Thus, the graph created in the root has $\sum_i n_i$
    vertices, and $O(2^d n^{1-1/d} \log^2 n )$ edges. Running \Boruvka
    algorithm on this graph takes $O(n^{1-1/d} \log^3 n )$ time. We
    thus get the recurrence
    \begin{equation*}
        T(n) = O(n^{1-1/d} \log^3 n ) + \sum_i T(n_i).
    \end{equation*}
    It is easy to verify that the solution to this recurrence is
    $O(n)$, as $\sum_i n_i = n$ and $n_i < n/2$ with high
    probability. (To convince yourself of this, consider the
    over-simplified recurrence $S(n) = O(n^{1-1/d}) + 2^dS(n/2^d)$.)
\end{proof}

\begin{remark}
    Note that the linear time \MST algorithm can also be extended to a
    linear time \MST algorithm for graphs with small separators. In
    that case, the portals are the separator vertices in the separator
    hierarchy, and we run the restricted \Boruvka bottom up on the
    separator decomposition tree.
\end{remark}

\subsection{The result}

The details of the following results are described in \secref{m_s_t}.

\begin{theorem}
    For fixed constant $d$, the \MST of $n$ uniformly and
    independently sampled points from $ [0,1]^d$ can be computed, by
    the above algorithm, in $O(n)$ expected time.
\end{theorem}

\section{Simple distance selection in %
   \TPDF{$O(n^{4/3}\log^{2/3}n)$}{n to the 4/3}
   time in \TPDF{$d=2$}{d=2}}

\myparagraph{The task} The input is a set $\P$ of $n$ points picked
randomly in $[0,1]^2$. For two sets $X,Y$, let
\begin{equation*}
    X \aotimes Y = \Set{ \bigl \{x,y\}}{ x \in X, y\in Y, x \neq y}
\end{equation*}
be the set of all unordered pairs in $X \times Y$.  Let
$\Pairs = \P \aotimes \P $, and for a fixed radius $\lambda$, let
$\Pairs_{\leq r} = \Set{ \{\p, \q\} \in \Pairs}{\dY{\p}{\q} \leq r}$
be the number of all pairs in $\P$ that are in distance at most $r$
from each other.  The task at hand is to compute
$\cardin{\Pairs_{\leq r}}$.

\myparagraph{Basic idea and some tools}

Let $\grid$ be a uniform $\nG \times \nG$ grid $\grid$, where
$\nG = \ceil{ (n / \log n)^{1/3}}$.  Let $\P_{i,j}$ denote the points
of $\P$ that fall in the grid cell
$\cell_{i,j} = [i/\nG,(i+1)/\nG] \times [j/\nG, (j+1)/\nG]$. Let
$\dc = \diamX{\cell_{i,j}} = \sqrt{2}/ \nG$ be the diameter of a grid
cells. We assume here that $r > 8 \dc$. The case for $r\leq 8\dc$ can be handled simply by bruteforce search of a fine grid.

Let
$\pr_{i,j} = \cardin{\Set{ \p \q \in \Pairs_{\leq r}}{ \p \in
      \P_{i,j}}}$. Observe that
$\cardin{\Pairs_{\leq r}} = \sum_{i,j} |\pr_{i,j}|/2 $. Thus, we
restrict our attention to computing the values of $\pr_{i,j}$, for all
$i,j$.  For a grid cell $\cell \in \grid$, consider the sets
\begin{equation*}
    \ballZY{\grid}{\cell}
    =%
    \Set{ \cellA \in \grid}{ \cellA \subseteq
       \ballY{\cenX{\cell}}{r-2\dc} \bigr.}
    \quad\text{and}\quad
    \BallY{\grid}{\cell}
    =%
    \Set{ \cellA \in \grid}{ \cellA \cap
       \ballY{\cenX{\cell}}{r+2\dc} \neq \emptyset \bigr.},
\end{equation*}
where $\cenX{\cell}$ is the \emph{center} of $\cell$.  All the grid
cells of $\ballZY{\grid}{\cell}$ are contained in any disk of radius
$r$ centered at a point of $\cell$.  Similarly,
$ \BallY{\grid}{\cell}$ is a super set of all the grid cells that
cover any disk of radius $r$ centered at any point of $\cell$.

Let
$\alpha_{i,j} = \cardin{(\ballZY{\grid}{\cell_{i,j}}\cap \P) \aotimes
   \P_{i,j}}$ and
$\beta_{i,j} = \cardin{(\BallY{\grid}{\cell_{i,j}}\cap \P) \aotimes
   \P_{i,j}}$. Observe that
$\alpha_{i,j } \leq \pr_{i,j} \leq \beta_{i,j}$.  The set
$\Ring_{i,j} = \BallY{\grid}{\cell_{i,j}} \setminus
\ballZY{\grid}{\cell_{i,j}}$ is formed by all the grid cells
intersecting a ring with outer radius $r+2\dc$ and inner radius
$r-2\dc$. Let $\PA_{i,j} = (\cup\Ring_{i,j}) \cap \P_{i,j}$.  Observe
that $\P_{i,j}$ and $\PA_{i,j}$ are disjoint. Consider the set of
pairs they induce $\P_{i,j} \aotimes \PA_{i,j}$, and let $\tau_{i,j}$
be the number of pairs in $\P_{i,j} \aotimes \PA_{i,j}$ of length at
most $r$. We have that $\pr_{i,j} = \alpha_{i,j} + \tau_{i,j}$. Thus,
the algorithm would compute the quantities $\alpha_{i,j}$ and
$\tau_{i,j}$ for all $i,j$. The algorithm would then compute
$\sum_{i,j} \pr_{i,j}/2$, which is the desired quantity.

\myparagraph{Low level procedures} In the following, we assume that
$n_{i,j} = |\P_{i,j}| = O(n/\nG^2)$.

\begin{lemma}
    After $O(n + \nG^2)$ preprocessing, given a query of numbers
    $i,j$, one can compute $\alpha_{i,j}$ in $O( \nG)$ time.
\end{lemma}
\begin{proof}
    The algorithm computes the grid $\grid$, the subset of points in
    $\P$ in each grid cell, and their number. The algorithm then
    preprocess the grid so that given an a contiguous range of cells
    in a row (of the grid), the algorithm can report the number of
    points in this range in $O(1)$ time. This can be done using prefix
    sums for each row of the grid.

    The desired quantity is
    $\alpha_{i,j} = \cardin{(\ballZY{\grid}{\cell_{i,j}}\cap \P)
       \aotimes \P_{i,j}} = n_{i,j} \sum_{\cell_{u,v} \in
       \ballZY{\grid}{\cell_{i,j}}} n_{u,v} - n_{i,j}^2 +
    \binom{n_{i,j}}{2} $.  The set $\ballZY{\grid}{\cell}$ in a row
    (of the grid) is just an contiguous box, and one can compute the
    number of points of $\P$ inside this box in $O(1)$ time. Thus,
    computing
    $\sum_{\cell_{u,v} \in \ballZY{\grid}{\cell_{i,j}}} n_{u,v}$ can
    be done in $O( \nG)$ time.
\end{proof}

\begin{lemma}
    \lemlab{p:a:i:j}%
    After $O(n + \nG^2)$ preprocessing, given a query numbers $i,j$,
    one can compute the set
    $\PA_{i,j} = (\cup \Ring_{i,j}) \cap \P_{i,j}$ in $O( n / \nG)$
    time (this also bounds its size).
\end{lemma}
\begin{proof}
    The set $\Ring$ is a ``ring'' of the grid of with $4$, and thus
    $|\Ring_{i,j}| =O(\nG)$. In particular, the set $\Ring_{i,j}$ can
    be computed in $O( \nG)$ time. The set $\PA_{i,j}$ is formed by
    collecting all the point sets $\P_{i,j}$ for cells
    $\cell_{i,j} \in \Ring_{i,j}$. By assumption,
    $|\P_{i,j}| =O( n /\nG^2)$, which readily implies that
    $|\PA_{i,j}| = O( \nG \cdot n /\nG^2 ) = O(n /\nG)$
\end{proof}

\begin{lemma}
    \lemlab{count:p:q}%
    Let $\PA$ and $\PB$ be two disjoint point sets in the plane, with
    $|\PA| < |\PB|$. Then one can compute the number of pairs of
    points in $\PA \aotimes \PB$ that are in distance at most $r$ from
    each other in $O( |\PA|^2 + |\PB| \log |\PA|)$ time.
\end{lemma}
\begin{proof}
    Let $\D$ be the set of disks of radius $r$ centered at the points
    of $\PA$. Compute the arrangement $\Arr = \ArrX{\PA}$, and compute
    for every face of $\Arr$ how many disks of $\D$ contain
    it. Furthermore, preprocess this arrangement for point-location
    queries in logarithmic time. This is all standard, and can be done
    in $O( |\PA|^2)$ time \cite{bcko-cgaa-08}. Now compute for each
    point of $\PB$ how many points of $\PA$ are in distance at most
    $r$ from it, by performing a point-location query in $\Arr$, and
    returning the depth of the query point.
\end{proof}

\myparagraph{Algorithm restated}

The algorithm computes $\alpha_{i,j}, \P_{i,j}, \PA_{i,j}$ for all
$i,j$ using the above procedures. It then computes for all $i,j$, the
quantity $\tau_{i,j}$ by using \lemref{count:p:q}. The algorithm now
computes directly $\sum_{i,j} (\alpha_{i,j} + \tau_{i,j})/2$ and
return it as the desired quantity.

\myparagraph{Analysis}
\begin{lemma}
    \lemlab{fewpoints-pergrid}
    Assuming $\nG = O( \sqrt{n} / \log n)$, with probability
    $\geq 1-1/n^{O(1)}$, each grid cell contains $O(n / \nG^2)$ points
    of the random point set $\P$.
\end{lemma}
\begin{proof}
    Each grid cell in the grid, in expectation, has
    $n / \nG^2 = \Omega( \log^2 n)$ points of $\P$ in it. Now using
    Chernoff's inequality it follows that this quantity is
    concentrated (say up to $1\pm 1/2$ around its expectation) with
    probability $\geq 1-1/n^{O(1)}$. Using the union bound on the
    $\nG^2$ grid cells, imply the claim.
\end{proof}

\myparagraph{Running time analysis}

Computing the sets $\PA_{i,j}$, for all $i,j \in \IRX{\nG}$, takes
$O( n \nG)$ time, using \lemref{p:a:i:j}. Computing $\tau_{i,j}$,
using \lemref{count:p:q}, takes
\begin{equation*}
    O( |\P_{i,j}|^2 + |\PA_{i,j}| \log |\P_{i,j}|) = O\pth{
       (n/\nG^2)^2 + (n/\nG)\log n}
\end{equation*}
time. doing this for all $i,j \in \IRX{\nG}$ takes
$ O\pth{ n^2/\nG^2 + n \nG\log n }$ time. Clearly, this dominates the
running time. Solving for $n^2/\nG^2 = n \nG\log n $, we get
$\nG = (n / \log n)^{1/3}$. Clearly, the last step dominates the
overall running time, which is
$o(n \nG\log n ) = O( n^{4/3} \log^{2/3} n)$.

\begin{theorem}
    Let $\P$ be a set of $n$ points picked uniformly and independently
    from $[0,1]^2$, and let $r$ be a parameter. One can compute, using
    the algorithm described above, the number of pairs of points in
    $\P$ in distance $\leq r$ from each other, in
    $O( n^{4/3} \log^{2/3} n)$ time. The result returned by the
    algorithm is always correct, and the bound on the running time
    holds with probability $\geq 1 -1 /n^{O(1)}$.
\end{theorem}

    \remove{%
       For example, plugging $\tau = 100\sqrt{c_d }n\log n = \tldO(n)$
       into \lemref{kutin-mcdiarmid} yields (for sufficiently large
       $c_d$):
    \[
        \Prob{\cardin{f_r(\P)-\Ex{f_r(\P)}} > 100\sqrt{c_dd}n\log n}
        \leq%
        2\bigl(\exp{\bigl(-\frac{10000c_d dn^2 \log(n)^2
           }{2dn(2\sqrt{c_d n \log n}+1)^2}\bigr)}+ \frac{dn^2}{dn^3}
        \bigr) \leq \frac{4}{n}.
    \]}

\section{Conclusions}
To get \Boruvka's algorithm to run in $O(n)$ time for \MST, we had to
restrict its growth phase in each recursive call. This feels unnatural
in many ways since it is intentionally slowing down the algorithm's
progress, but is necessary for a complete analysis. It remains open
whether there is a method of showing \Boruvka algorithm takes linear
time in three or higher dimensions on random points. One possible
direction would be to show that the average degree of the connected
components in $\GCY{\P}{\vlr}$, see \defref{g:c:y}, increases (for
$d\geq 3$) extremely slowly compared to the halving of connected
components. This is an observation the authors noted in numerical
simulations, yet were unable to prove. See
\figref{degree_distribution_boruvka}. If the average degree increase
in every round of \Boruvka's algorithm can be bounded to a
multiplicative constant $\xi < 2$ in each round then that would imply
that \Boruvka's algorithm runs in linear time.

\SoCGVer{%
   \bibliographystyle{plain} %
   \bibliography{rand_mst} %
}

\NotSoCGVer{%
   \BibTexMode{%
      \bibliographystyle{alpha} \bibliography{rand_mst} }%
   \BibLatexMode{\printbibliography} }

\appendix

\AppendixOfProofs

\end{document}